\documentclass[draft,narroweqnarray,inline]{ieee}
\pdfoutput=1
\usepackage{graphicx}
\usepackage{psfrag}
\usepackage{amsmath,amssymb}

\def\classP{$\boldsymbol P$ }
\def\classNP{$\boldsymbol N \! \boldsymbol P$ }

\newtheorem{theorem}{Theorem}[section]

\newtheorem{corollary}{{Corollary}}[section]
\def\len{\mbox{len}}

\begin{document}


\title[A 3-CNF-SAT descriptor algebra and the solution of the
\classP=\classNP conjecture]{A
  3-CNF-SAT descriptor algebra and  the solution of the \classP=\classNP conjecture}

\author[M.R\'emon & J.Barth\'elemy]{%
      Prof. Marcel R\'emon\authorinfo{%
M.R\'emon, Department of Mathematics,
      Namur University, Belgium;
        \mbox{marcel.remon@unamur.be}}
\and and Dr. Johan Barth\'elemy\authorinfo{%
J.Barth\'elemy, SMART Infrastructure Facilities, University of Wollongong, Australia; \mbox{johan@uow.edu.au}} 
      }

\maketitle

\begin{abstract}
\noindent The relationship between the complexity classes \classP and
\classNP is an unsolved question in the field of theoretical computer
science. In this paper, we investigate a descriptor approach based on
lattice properties.  In a previous paper, we tried to prove that
neither \classP $\neq$ \classNP nor \classP = \classNP was
``unprovable'' within the  a-temporal framework of Mathematics.  See
\cite{arxiv2}. A part of the proof about the impossibility to prove
that \classP = \classNP turns to be inexact, and yields the first author to
investigate deeper into the possibility of \classP = \classNP.\\
 This
paper proposes a new way to decide the satisfiability
of any 3-CNF-SAT problem.  The analysis of this
exact [non
heuristical] algorithm shows {\bf \em a strictly bounded exponential
  complexity}.  The complexity of any 3-CNF-SAT solution is bounded by
${\cal O}(2^{490})$. This [over-estimated] bound is reached by an algorithm working on the
smallest description (via descriptor functions) of the evolving
set of solutions in function of the already considered clauses, without exploring these solutions.  Any remark about this paper is warmly welcome.
\end{abstract}




\section{The 3-CNF-SAT problem}
\noindent {\it Boolean formulae} are built in the usual way from
propositional variables $x_i$ and the logical connectives $\wedge$, $\vee$ and $\neg$, which are interpreted as conjunction, disjunction, and negation, respectively.  A {\it literal} is a propositional variable or the negation of a propositional variable, and a {\it clause} is a disjunction of literals.  A Boolean formula is {\it in conjunctive normal form} if and only if it is a conjunction of clauses. \\[12pt]
\noindent A {\it 3-CNF formula} $\varphi$ is a Boolean formula in
conjunctive normal form with exactly three literals per clause, like
$\varphi := (x_1 \vee x_2 \vee \neg x_3) \wedge (\neg x_2 \vee  x_3
\vee \neg x_4):= \psi_1 \wedge \psi_2 $.  A {\it 3-CNF formula} is
composed of $n$ propositional variables $x_i$ and $m$ clauses
$\psi_j$.  \\[12pt]
\noindent The {\it 3-CNF-satisfiability or 3-CNF-SAT problem} is to
decide whether there exists or not logical values for the
propositional variables, so that $\varphi$ can be true.  Until now, we
do not know whether it is possible or not to
check the satisfiability of any given {\it 3-CNF} formula $\varphi$ in
a polynomial time with respect of $n$, as the {\it 3-CNF-SAT} problem is known to belong to the class \classNP of problems. See \cite{cormen2001} for details. 
\section{A matrix representation of a 3-CNF formula}
\subsection{Definitions}
\noindent The {\it size} of a 3-CNF formula $\varphi$ is defined
as the size of the corresponding {\it Boolean circuit}, i.e. the
number of logical connectives in  $\varphi$.  Let us note
the following property :
\begin{eqnarray}
 \mbox{\it size($\varphi$)} = {\cal O}(m) = {\cal O}(\alpha \times n)
\label{no1}
\end{eqnarray}
\noindent where $\alpha = m / n $ is the {\it ratio} of
clauses with respect to variables. It seems that $\alpha \approx 4.258$ 
gives the most difficult 3-CNF-SAT problems.  See \cite{Crawford199631}. \\[12pt] 
\noindent Let $\varphi(x_1,x_2,\cdots,x_n)$ be a 3-CNF formula. 
 The set ${\cal S}_{\varphi}$ of all {\it satisfying} solutions is 
\begin{eqnarray}
 {\cal S}_{\varphi} = \{ (x_1,\cdots,x_n) \in
\{0,1\}^n \;  | \; \varphi(x_1,\cdots,x_n) =1 \} 
\label{no2}
\end{eqnarray}
Let $\Sigma_{\varphi} = \# \; {\cal S}_{\varphi}$ and $\bar{s}_1, \cdots,
\bar{s}_{\Sigma_{\varphi}}$ be the ordered elements of $ {\cal
  S}_{\varphi}$.  For $1 \leq j \leq \Sigma_{\varphi} : \bar{s}_j =
(s_j^1,\cdots,s_j^i,\cdots,s_j^n)$. We define the {\it 3-CNF-matrix representation} of $\varphi$ as
$[\varphi]$ :
\begin{eqnarray}
 [\varphi] =  \left( \begin{array}{ccc}
x_1 & x_i & x_n \\
\hline
s_1^1 & \cdots & s_1^n \\
\vdots & s_j^i & \vdots  \\
s_{\Sigma_{\varphi}}^1& \cdots & s_{\Sigma_{\varphi}}^n \end{array}
\right)
\label{no3}
\end{eqnarray}
\subsection{Examples}
\noindent Each clause $\psi_i$ will be represented by a $7 \times 3$
matrix. For example,
\begin{eqnarray*}
 [\psi_1] = [(x_1 \vee x_2 \vee \neg x_3)] =
 \left( \begin{array}{ccc}
x_1 & x_2 & x_3 \\
\hline
0 & 0 & 0 \\
0 & 1 & 0 \\
0 & 1 & 1 \\
1 & 0 & 0 \\
1 & 0 & 1 \\
1 & 1 & 0 \\
1 & 1 & 1 \\
\end{array} \right) \mbox{  and  } 
 [\psi_2] = [(\neg x_2 \vee  x_3
\vee \neg x_4)] =
 \left( \begin{array}{ccc}
x_2 & x_3 & x_4 \\
\hline
0 & 0 & 0 \\
0 & 0 & 1 \\
0 & 1 & 0 \\
0 & 1 & 1 \\
1 & 0 & 0 \\
1 & 1 & 0 \\
1 & 1 & 1 \\
\end{array} \right)  
\end{eqnarray*}
The 3-CNF formula $\varphi = \psi_1 \wedge \psi_2$ will be represented by a $12
\times 4$ matrix :
\begin{eqnarray*}
 [\varphi] = [(x_1 \vee x_2 \vee \neg x_3) \wedge (\neg x_2 \vee  x_3
\vee \neg x_4)] =
 \left( \begin{array}{cccc}
x_1 & x_2 & x_3 & x_4\\
\hline
0 & 0 & 0 & 0 \\
0 & 0 & 0 & 1 \\
0 & 1 & 0 & 0 \\
0 & 1 & 1 & 0 \\
0 & 1 & 1 & 1 \\
1 & 0 & 0 & 0 \\
1 & 0 & 0 & 1 \\
1 & 0 & 1 & 0 \\
1 & 0 & 1 & 1 \\
1 & 1 & 0 & 0 \\
1 & 1 & 1 & 0 \\
1 & 1 & 1 & 1 \\
\end{array} \right)  
\end{eqnarray*}
This paper defines an algebra on this type of matrices such
that $[\varphi] = [\psi_1] \wedge [\psi_2] $.
\section{First definitions and properties for 3-CNF-matrices}
\subsection{Extension to new variables}
\noindent Let $A$ be such a matrix, $A$ can be {\it extended} to new
propositional variables by adding columns filled with the neutral sign
``.'', meaning that the corresponding variable can be set either to 0 or 1. This new matrix $\overline{A}$ is equivalent to $A$.
\begin{eqnarray}
 A = 
 \left( \begin{array}{ccc}
x_1 & x_2 & x_4\\
\hline
a_1^1 & a_1^2 & a_1^4 \\ 
a_j^1 & a_j^2 & a_j^4 \\ 
a_{\Sigma_{\varphi}}^1 & a_{\Sigma_{\varphi}}^2 & a_{\Sigma_{\varphi}}^4 \\ 
\end{array} \right)
\equiv 
\left( \begin{array}{cccc}
x_1 & x_2 & x_3 & x_4\\
\hline
a_1^1 & a_1^2 & . [_1^0] & a_1^4 \\ 
a_j^1 & a_j^2 & . & a_j^4 \\ 
a_{\Sigma_{\varphi}}^1 & a_{\Sigma_{\varphi}}^2 & . & a_{\Sigma_{\varphi}}^4 \\ 
\end{array} \right)=  \overline{A}
\end{eqnarray}
\subsection{Reduction of 3-CNF-matrices}
\noindent The inverse operation, called {\it reduction}, replaces
two same lines only differing by a 0 and a 1 for a variable, by a
unique line with a neutral sign for this variable : 
\begin{eqnarray}
 A = 
 \left( \begin{array}{ccc}
x_1 & x_2 & x_3 \\
\hline
0 & 0 & 0 \\
0 & 0 & 1 \\
0 & 1 & 0 \\
0 & 1 & 1 \\
1 & 0 & 0 \\
1 & 0 & 1 \\
1 & 1 & 0 \\
\end{array} \right) \equiv
 \left( \begin{array}{ccc}
x_1 & x_2 & x_3 \\
\hline
0 & 0 & . \\
0 & 1 & . \\
1 & 0 & . \\
1 & 1 & 0 \\
\end{array} \right) \equiv
 \left( \begin{array}{ccc}
x_1 & x_2 & x_3 \\
\hline
0 & . & . \\
1 & 0 & . \\
1 & 1 & 0 \\
\end{array} \right)  \label{no5}
\end{eqnarray}

\subsection{Disjunction of 3-CNF-matrices}
\noindent Let $A$ and $B$ be two matrices and $\{x_1, \cdots, x_n\}$ the
union of their support variables.  Let $\overline{A}$ and $\overline{B}$ be their
extensions over $\{x_1, \cdots, x_n\}$.  Then we define the {\it
  disjunction} of $A$ and $B$ by
\begin{eqnarray}
 A \vee B = 
 \left( \begin{array}{c}
x_1 \; \cdots \; \; x_n\\
\hline
\overline{A} \\ 
\overline{B} \\ 
\end{array} \right)
\end{eqnarray}
Of course, this new matrix should be reordered so that the lines are
in a ascending order, which can yield sometimes in replacing a line with
a neutral sign by two lines with a one and a zero.
\subsection{Block decomposition of 3-CNF-matrices} 
\noindent Let $A$ a matrix such that the {\it reduction} process
yields to lines with neutral sign, then $A$ can be rewritten as the
disjunction of smaller matrices.  For example,  
\begin{eqnarray*}
 [\psi_1] = 
 \left( \begin{array}{ccc}
x_1 & x_2 & x_3 \\
\hline
0 & 0 & 0 \\
0 & 1 & 0 \\
0 & 1 & 1 \\
1 & 0 & 0 \\
1 & 0 & 1 \\
1 & 1 & 0 \\
1 & 1 & 1 \\
\end{array} \right) = 
 \left( \begin{array}{c}
x_1  \\
\hline
1 \\
\end{array} \right) \vee
 \left( \begin{array}{cc}
x_1 & x_2 \\
\hline
0 & 1 \\
\end{array} \right) \vee
 \left( \begin{array}{ccc}
x_1 & x_2 & x_3 \\
\hline
0 & 0 & 0 \\
\end{array} \right)
\end{eqnarray*}
The block decomposition is not unique. For instance, there are 6 different block
decompositions for a 3-variables clause.
\subsection{Conjunction of 3-CNF-matrices}
\noindent Let $A$ and $B$ be two matrices, $\overline{A}$ and $\overline{B}$
their extensions to the joint set of propositional variables. 
Let $\overline{A}_k$ and $\overline{B}_l$ be the {\it one line matrices}
such that :
\begin{eqnarray}
\overline{A} =
\underset{k=1}{\overset{\Sigma_{\overline{A}}}{\bigvee}}
\overline{A}_k \mbox{ and }
\overline{B} =
\underset{l=1}{\overset{\Sigma_{\overline{B}}}{\bigvee}}
\overline{B}_l 
\end{eqnarray}
 We
define {\it the conjunction} of $A$ and $B$ as 
\begin{eqnarray}
 A \wedge B \equiv \overline{A} \wedge \overline{B} = 
\left( \underset{k=1}{\overset{\Sigma_{\overline{A}}}{\bigvee}}
\overline{A}_k \right) \wedge 
\left( 
\underset{l=1}{\overset{\Sigma_{\overline{B}}}{\bigvee}}
\overline{B}_l \right)=
\underset{k=1}{\overset{\Sigma_{\overline{A}}}{\bigvee}} \;
\underset{l=1}{\overset{\Sigma_{\overline{B}}}{\bigvee}}
\left( \overline{A}_k \wedge \overline{B}_l \right) = 
\underset{k=1}{\overset{\Sigma_{\overline{A}}}{\bigvee}} \;
\underset{l=1}{\overset{\Sigma_{\overline{B}}}{\bigvee}}
\overline{C}_{k,l} 
\end{eqnarray}
where 
\begin{eqnarray}
\overline{C}_{k,l} = 
\left( \begin{array}{ccc}
x_1 & x_i & x_n \\
\hline
a_k^1 & a_k^i & a_k^n \\
\end{array}
\right) \wedge 
\left( \begin{array}{ccc}
x_1 & x_i & x_n \\
\hline
b_l^1 & b_l^i & b_l^n \\
\end{array}
\right)
= \left\{ \begin{array}{l}
\; \emptyset \mbox{ if } \exists \; c_m^i = \mbox{\it ``NaN"}\\
\left( \begin{array}{ccc}
x_1 & x_i & x_n \\
\hline
c_m^1 & c_m^i & c_m^n \\
\end{array}
\right) \mbox{ otherwise }
\end{array} \right.
\end{eqnarray}
and
\begin{eqnarray}
c_m^i = 
\left\{ \begin{array}{l}
a_k^i \mbox{ if } a_k^i = b_l^i \\
a_k^i \mbox{ if } a_k^i \neq b_l^i \mbox{ and } b_l^i = ``\cdot"  \\
b_l^i \mbox{ if } a_k^i \neq b_l^i \mbox{ and } a_k^i = ``\cdot"  \\
\mbox{\it ``NaN"} \mbox{ otherwise} 
\end{array} \right.
\end{eqnarray}
\subsection{The empty and full 3-CNF-matrices}
\noindent Let us call $\emptyset$, the {\it empty matrix}, with no
line at all.  The empty matrix is neutral for the disjunction operator
$\vee$ and absorbing for the conjunction operator $\wedge$.\\[12pt]
\noindent Let us define $\Omega$, the {\it full matrix}, as a one line matrix
with only neutral signs in it. The full matrix is neutral for $\wedge$ and absorbing for $\vee$.
\subsection{Example of operations}
Let us consider the following block decompositions for $[\psi_1]$ and $[\psi_2]$ with
$x_2$ and $(x_2 \; x_3)$ as common supports.
\begin{eqnarray*}
 [\psi_1] = 
 \left( \begin{array}{ccc}
x_1 & x_2 & x_3 \\
\hline
0 & 0 & 0 \\
0 & 1 & 0 \\
0 & 1 & 1 \\
1 & 0 & 0 \\
1 & 0 & 1 \\
1 & 1 & 0 \\
1 & 1 & 1 \\
\end{array} \right) = 
 \left( \begin{array}{c}
x_2  \\
\hline
1 \\
\end{array} \right) \vee
 \left( \begin{array}{cc}
x_2 & x_3 \\
\hline
0 & 0 \\
\end{array} \right) \vee
 \left( \begin{array}{ccc}
x_1 & x_2 & x_3 \\
\hline
1 & 0 & 1 \\
\end{array} \right)
\end{eqnarray*}
and
\begin{eqnarray*}
 [\psi_2] = 
 \left( \begin{array}{ccc}
x_2 & x_3 & x_4 \\
\hline
0 & 0 & 0 \\
0 & 0 & 1 \\
0 & 1 & 0 \\
0 & 1 & 1 \\
1 & 0 & 0 \\
1 & 1 & 0 \\
1 & 1 & 1 \\
\end{array} \right) = 
 \left( \begin{array}{c}
x_2  \\
\hline
0 \\
\end{array} \right) \vee
 \left( \begin{array}{cc}
x_2 & x_3 \\
\hline
1 & 1 \\
\end{array} \right) \vee
 \left( \begin{array}{ccc}
x_2 & x_3 & x_4 \\
\hline
1 & 0 & 0 \\
\end{array} \right)
\end{eqnarray*}
\begin{flalign*}
& [\psi_1] \wedge [\psi_2] &\\
& \hspace{15pt} = \emptyset\vee
 \left( \begin{array}{cc}
x_2 & x_3 \\
\hline
1 & 1 \\
\end{array} \right) \vee
 \left( \begin{array}{ccc}
x_2 & x_3 & x_4 \\
\hline
1 & 0 & 0 \\
\end{array} \right) \vee
 \left( \begin{array}{cc}
x_2 & x_3 \\
\hline
0 & 0 \\
\end{array} \right) \vee
\emptyset \vee
\emptyset\vee
 \left( \begin{array}{ccc}
x_1 & x_2 & x_3 \\
\hline
1 & 0 & 1 \\
\end{array} \right) \vee
\emptyset \vee
\emptyset & \\
& \hspace{15pt} = \left( \begin{array}{cccc}
x_1 & x_2 & x_3 & x_4\\
\hline
. & 0 & 0 & . \\ 
. & 1 & 0 & 0 \\ 
. & 1 & 1 & . \\
1 & 0 & 1 & . \\ 
\end{array} \right)&
\end{flalign*}
\mbox{}\\[3pt]
\subsection{Lattice structure of 3-CNF-matrices}
\noindent A {\it semi-lattice} $(X,\vee)$ is a pair consisting of a set X
and a binary operation $\vee$ which is associative, commutative, and
idempotent. \\[12pt]
\noindent Let us note ${\cal A}$ the set of all the
3-CNF-matrices. Then $({\cal A},\vee)$ and $({\cal A},\wedge)$ are
both semi-lattices, respectively called {\it join} and {\it meet}
semi-lattices. \\[12pt] 
\noindent Let us define {\it the two absorption laws} as  $x = x \vee
(x \wedge y)$  and its dual $x = x \wedge (x \vee y)$.
A {\it lattice} is an algebra $(X, \vee, \wedge)$ satisfying equations
expressing associativity, commutativity, and idempotence of $\vee$ and
$\wedge$, and satisfying the two absorption equations. \\[12pt]
\noindent $({\cal A},\vee, \wedge)$ is a lattice over the set of
3-CNF-matrices with respect to the disjunction and conjunction
operators.  Moreover, $({\cal A},\vee, \wedge)$ is {\bf \em a
  distributive bounded
lattice} as $\wedge$ is distributive with respect to $\vee$ and $A
\vee \Omega = \Omega \; \; \& \;\; A \wedge \emptyset = \emptyset \;\; \forall A
\in {\cal A}$.  See \cite{burris2012a} for more details over lattices. 

\section{Characterization theorems via functional descriptors}
\begin{theorem}{\bf Every non empty 3-CNF-matrix can be characterized by a
    one-line parameterized 3-CNF-matrix, called its functional matrix
    description.} 
\begin{eqnarray}
\forall \;  [\varphi] &=&  \left( \begin{array}{ccc}
x_1 & x_i & x_n \\
\hline
s_1^1 & \cdots & s_1^n \\
\vdots & s_j^i & \vdots  \\
s_{\Sigma_{\varphi}}^1& \cdots & s_{\Sigma_{\varphi}}^n \end{array}
\right)\neq \emptyset
\; , \; \exists \; n \mbox{ functions } f_i : \{0,1\}^i
\rightarrow \{0,1\} \mbox{ such that } \nonumber \\[12pt]
\mbox{}
[\varphi] &=&  \underset{(\alpha_1,\cdots,\alpha_n) \in \{0,1\}^n}{\bigvee}
 \left( \begin{array}{ccccc}
x_1 & \cdots & x_i & \cdots & x_n \\
\hline
f_1(\alpha_1) & \cdots & f_i(\alpha_1,\cdots,\alpha_i) & \cdots &
f_n(\alpha_1, \cdots, \alpha_n) \\
\end{array}
\right) \label{h-def}\\
\mbox{} \nonumber\\
&\stackrel{notation}{\equiv}& \left[ \begin{array}{c}
f_1(\alpha_1)\\
\vdots\\
f_n(\alpha_1,\cdots,\alpha_n)
\end{array}
\right] \label{no12} 
\end{eqnarray}
\end{theorem}
\mbox{}\\[12pt]
\noindent So, the knowledge of $f_1(\alpha_1),  \cdots,  f_i(\alpha_1,\cdots,\alpha_i),  \cdots, 
f_n(\alpha_1, \cdots, \alpha_n)$ characterizes fully $ [\varphi] $.
These modulo-2 functions are called {\bf the functional descriptors}
of $\varphi$.
\mbox{}\\[12pt]
\noindent {\it Example : } 
\begin{flalign*}
 [\varphi] & = [(x_1 \vee x_2 \vee \neg x_3) \wedge (\neg x_2 \vee  x_3
\vee \neg x_4)] \\
& = \underset{(\alpha_1,\cdots,\alpha_4) \in
  \{0,1\}^4}{\bigvee}  
 \left( \begin{array}{cccc}
x_1 & x_2 & x_3 & x_4 \\
\hline
\alpha_1 & \alpha_2 & (\alpha_1+1)(\alpha_2+1)\alpha_3+\alpha_3 &
\alpha_2 (\alpha_3+1) \alpha_4 +\alpha_4\\
\end{array}
\right)_{ \; \; \mbox{(mod 2)}}
\end{flalign*}
\mbox{}\\[12pt]
\begin{proof}
\mbox{}\\[12pt]
\noindent $\bullet$ The theorem is satisfied for $n=1$ as 
\begin{eqnarray*}
\left( \begin{array}{c}
x_1 \\
\hline
1  \\
\end{array} \right)
= 
\left( \begin{array}{c}
x_1 \\
\hline
f_1(\alpha_1) \equiv 1  \\
\end{array} \right) \; ; \; 
\left( \begin{array}{c}
x_1 \\
\hline
0  \\
\end{array} \right)
= 
\left( \begin{array}{c}
x_1 \\
\hline
f_1(\alpha_1) \equiv 0  \\
\end{array} \right) \; ; \; 
\left( \begin{array}{c}
x_1 \\
\hline
0  \\
1 \\
\end{array} \right)
= 
\underset{\alpha_1 \in
  \{0,1\}}{\bigvee}
 \left( \begin{array}{c}
x_1 \\
\hline
\alpha_1   \\ 
\end{array} \right)
\end{eqnarray*}
$\bullet$
Let the theorem be true for $n-1$ and $[\varphi]$ be a 3-CNF-matrix of
dimension $n$.  There exist two 3-CNF-matrices $[\varphi_1]$ and
$[\varphi_2]$ of size $n-1$ such that :
\begin{eqnarray*}
[\varphi] = \underset{\alpha_i \in
  \{0,1\}}{\bigvee} \left( \begin{array}{cc}
x_1 & x_2 \cdots x_n \\
\hline
0 & f_2(\alpha_2) \cdots f_n(\alpha_2,\cdots,\alpha_n)   \\ 
\end{array} \right)
\underset{\alpha_i \in
  \{0,1\}}{\bigvee} \left( \begin{array}{cc}
x_1 & x_2 \cdots x_n \\
\hline
1 & g_2(\alpha_2) \cdots g_n(\alpha_2,\cdots,\alpha_n)   \\
\end{array} \right) \\
\end{eqnarray*}
\begin{eqnarray*}
\mbox{Thus \hspace{1cm}}  [\varphi] = \underset{\alpha_i \in
  \{0,1\}}{\bigvee} \left( \begin{array}{c}
 x_1 \cdots x_n \\
\hline
h_1(\alpha_1) \cdots h_n(\alpha_1,\cdots,\alpha_n)   \\ 
\end{array} \right) \mbox{\hspace{2,8cm}}
\end{eqnarray*}
\mbox{\hspace{1cm}} where 
\begin{eqnarray*}
 h_1(\alpha_1) & = & \alpha_1 \\
h_i(\alpha_1,\cdots,\alpha_i) & = & (\alpha_1+1) f_i(\alpha_2,\cdots,\alpha_i) + \alpha_1
g_i(\alpha_2,\cdots,\alpha_i) _{ \; \; \mbox{(mod 2)}} \; \; \mbox{ for } i \neq 1 
\end{eqnarray*}
\end{proof}
\begin{corollary}{\bf The functional descriptors of $\varphi$ are modulo-2 multi-linear
  combinations of $\alpha_i$.}
\end{corollary}

\begin{proof}
This is a mere consequence of the recursive definition of
$h_i(\alpha_1,\cdots,\alpha_i)$. 
\begin{eqnarray}
\mbox{ So, } h_i(\alpha_1,\cdots,\alpha_i)&=&\sum_{(\delta_1,\cdots,\delta_i)\in
  \{0,1\}^i} \Delta_i(\delta_1,\cdots,\delta_i) \;\; \alpha_1^{\delta_1}
\cdots \alpha_i^{\delta_i} \;_{ \; \; \mbox{(mod
    2)}}  \label{delta-def} \\
& &\mbox{\hspace{1,5cm} where } \;\; \Delta_i(\delta_1,\cdots,\delta_i)
\in \{0,1\} \nonumber \\
\Delta_i(\delta_1,\cdots,\delta_i) &\mbox{is}& \mbox{called the
  signature of } h_i(\alpha_1,\cdots,\alpha_i). \nonumber
\end{eqnarray}
\end{proof}
\noindent {\it Example : } \\
\noindent Consider a clause $\psi \equiv [\neg] x_r
\vee [\neg] x_s \vee [\neg] x_t$ where $1 \leq r < s < t \leq
n$. 
$[\psi]$ is then characterized by the following characterization functions : 
\begin{eqnarray}
h_i(\alpha_1, \cdots, \alpha_i) & = & \alpha_i \;\;\; \forall \; i < t
\nonumber \\
h_t(\alpha_r,\alpha_s,\alpha_t) & = & \left\{
\begin{array}{ll}
(\alpha_{r}+1)(\alpha_s+1)(\alpha_t+1)+\alpha_t \;\;  &
\mbox{if} \;\; \psi = x_r \vee x_s \vee x_t \\
(\alpha_{r}+1)(\alpha_s+1)\; \alpha_t+\alpha_t \;\;  &
\mbox{if} \;\; \psi = x_r \vee x_s \vee \neg x_t \\
(\alpha_{r}+1)\; \alpha_s \; (\alpha_t+1)+\alpha_t \;\;  &
\mbox{if} \;\; \psi = x_r \vee \neg  x_s \vee x_t \\
(\alpha_{r}+1)\; \alpha_s \; \alpha_t +\alpha_t \;\;  &
\mbox{if} \;\; \psi = x_r \vee \neg x_s \vee \neg x_t \\
\alpha_{r} \; (\alpha_s+1)(\alpha_t+1)+\alpha_t \;\;  &
\mbox{if} \;\; \psi = \neg x_r \vee x_s \vee x_t \\
\alpha_{r}\;(\alpha_s+1)\; \alpha_t+\alpha_t \;\;  &
\mbox{if} \;\; \psi = \neg x_r \vee x_s \vee \neg x_t \\
\alpha_{r}\; \alpha_s \; (\alpha_t+1)+\alpha_t \;\;  &
\mbox{if} \;\; \psi = \neg x_r \vee \neg  x_s \vee x_t \\
\alpha_{r}\; \alpha_s \; \alpha_t +\alpha_t \;\;  &
\mbox{if} \;\; \psi = \neg x_r \vee \neg x_s \vee \neg x_t \\
\end{array} \right. \label{def_h} \\
\nonumber
\end{eqnarray}
\begin{theorem} {\bf The conjunction operator $\wedge$ between two
  sets of clauses can be rewritten as the merging of their
  characterization functions : $[h_i(\cdot)] = [f_i(\cdot)] \wedge [g_i(\cdot)]$.}
\end{theorem}
%
\begin{eqnarray*} \mbox{ Let } [\varphi] \equiv \left[ \begin{array}{c}
f_1(\alpha_1)\\
\vdots\\
f_n(\alpha_1,\cdots,\alpha_n)
\end{array}
\right] \mbox{ and } [\varphi'] \equiv \left[ \begin{array}{c}
g_1(\alpha_1)\\
\vdots\\
g_n(\alpha_1,\cdots,\alpha_n)
\end{array}
\right] \\
\end{eqnarray*}
\noindent Note : $\varphi$ or $\varphi'$ should be extended if necessary
in order to get the same support of propositional variables.  Remember
that all operations are {\it modulo 2} : $\alpha_i + \alpha_i = 0$,
$\alpha_i^2 = \alpha_i$ and $(\alpha_i + 1) \alpha_i = 0$ for all $\alpha_i$.\\
\begin{eqnarray*} \mbox{ Then } [\varphi] \wedge [\varphi']  \equiv \left[ \begin{array}{c} 
h_1(\alpha_1)\\
\vdots\\ 
h_n(\alpha_1,\cdots,\alpha_n)
\end{array}
\right] \mbox{\hspace{2cm}}
\end{eqnarray*}
\noindent where for  $1 \leq t \leq n$ : 
\begin{eqnarray}
h_t(\alpha_1,\cdots, \alpha_t) 
= & (\alpha_t +1) &\; \cdot \;  \{ \; [f_t(\alpha_1,\cdots, \alpha_{t-1},0) +  
 g_t(\alpha_1,\cdots, \alpha_{t-1},0)] \label{merge} \\
& &  \; \; \;  \; \;\cdot \; [ f_t(\alpha_1,\cdots, \alpha_{t-1},1) \cdot  
 g_t(\alpha_1,\cdots, \alpha_{t-1},1)]  \nonumber \\
& &  \; \; + \; [f_t(\alpha_1,\cdots, \alpha_{t-1},0) \cdot  
 g_t(\alpha_1,\cdots, \alpha_{t-1},0)] \; \} \nonumber \\
&  \; +\;\;  \; \alpha_t &\; \cdot \;
   \{ \; [f_t(\alpha_1,\cdots, \alpha_{t-1},1) +  
 g_t(\alpha_1,\cdots, \alpha_{t-1},1)]  \nonumber \\
& &  \; \;\;  \; \;\cdot \; [f_t(\alpha_1,\cdots, \alpha_{t-1},0) +  
 g_t(\alpha_1,\cdots, \alpha_{t-1},0) ]  \nonumber \\
& &  \; \; \; + \; [f_t(\alpha_1,\cdots, \alpha_{t-1},1) +  
 g_t(\alpha_1,\cdots, \alpha_{t-1},1)]   \nonumber \\
& &  \;  \; \; \;  \; \cdot \; [f_t(\alpha_1,\cdots, \alpha_{t-1},0) \cdot  
 g_t(\alpha_1,\cdots, \alpha_{t-1},0)]   \nonumber \\
& & \;  \; \;+   \; [f_t(\alpha_1,\cdots, \alpha_{t-1},1) \cdot  
 g_t(\alpha_1,\cdots, \alpha_{t-1},1) ] \}\;_{ \; \; \mbox{(mod 2)}}  \nonumber
\end{eqnarray}
\begin{eqnarray}
\mbox{Moreover if there exists} & a & \mbox{(unique) $j < t$,
   related to the highest $\alpha_j$ such
   that :} \; \; \nonumber \\
  g^*_j(\alpha_1,\cdots, \alpha_j)& \equiv  & 
  [f_t(\alpha_1,\cdots, \alpha_{t-1},0) +  
 g_t(\alpha_1,\cdots, \alpha_{t-1},0)]  \; \cdot \; \nonumber \\
& & [ f_t(\alpha_1,\cdots, \alpha_{t-1},1) +  
 g_t(\alpha_1,\cdots, \alpha_{t-1},1)] \neq 0  \label{recursive} \\
& & \mbox{[An additional constraint over
  $\alpha_j$ induced by the conjunction operation]} \nonumber \\
& \Rightarrow &  \; \; \mbox{call to a new  merging } 
\; \; f_j(\alpha_1, \cdots, \alpha_j) \;\; \wedge
  \;\; g^*_j(\alpha_1, \cdots, \alpha_j) \nonumber \\
& & \; \; \mbox{using a recursive call to definition
  (\ref{merge})}. \nonumber \\ 
\end{eqnarray}
\noindent Recursivity will end as soon as there is no longer such
$g^*_j(\alpha_1,\cdots, \alpha_j) = 1$ or when $
g^*_j(\alpha_1,\cdots, \alpha_j)$ is no longer a function of
$\alpha_i$ but a constant always equal to $1$.\\
\mbox{} \\
\begin{proof}
Consider the possible values for 
$f_t(\alpha_1,\cdots,\alpha_t)$ and $g_t(\alpha_1,\cdots,\alpha_t)$ in
equation (\ref{merge}) when $\alpha_t \in \{0,1\}$ : \\
\begin{eqnarray*}
\bullet \; f_t(\alpha_1,\cdots,\alpha_{t})&=& g_t(\alpha_1,\cdots,\alpha_{t}) \;\;
\mbox{for } \alpha_t \in \{0,1\} \\ 
&\Downarrow & \\
h_t(\alpha_1,\cdots,\alpha_{t}) & = & 
(\alpha_t +1) \cdot \;  \{ \; [f_t(\cdot,0) +  g_t(\cdot,0)]  \; \cdot \; [ f_t(\cdot,1) \cdot  g_t(\cdot,1)] \; + \; [f_t(\cdot,0) \cdot  
 g_t(\cdot,0)] \; \}  \; + \\
& & \; \alpha_t \cdot   \{ \; [f_t(\cdot,1) +  
 g_t(\cdot,1)]  \; \cdot \; [f_t(\cdot,0) +  
 g_t(\cdot,0)  ]   \; + \\
& & \; \;\;\;\;\;\;\;\;\;\;\; [f_t(\cdot,1) +  
 g_t(\cdot,1)]  \; \cdot \; [f_t(\cdot,0) \cdot  
 g_t(\cdot,0)]  \; +   \; [f_t(\cdot,1) \cdot  
 g_t(\cdot,1) ] \;  \} \\
& = & (\alpha_t +1) \cdot  f_t(\cdot,0) +
 \alpha_t \cdot  f_t(\cdot,1)   \;\;\; \mbox{[as $f_t()+ g_t()=0$ and
   $f_t() \cdot g_t() = f_t^2()=f_t()$]} \\
& = & f_t(\alpha_1,\cdots,\alpha_{t}) =
h_t(\alpha_1,\cdots,\alpha_{t}) \\
& & \mbox{[$h_t()$ is thus the conjunction of $f_t()$ and $g_t()$]}
\\[24pt]
\bullet \; f_t(\alpha_1,\cdots,0) \;\; &=& g_t(\alpha_1,\cdots,0)  \;\; \mbox{ but } \;\;
f_t(\alpha_1,\cdots,1) \neq g_t(\alpha_1,\cdots,1)  \\ 
&\Downarrow & \\
h_t(\alpha_1,\cdots,\alpha_t) & = & (\alpha_t +1) \cdot  f_t(\cdot,0) +
 \alpha_t \cdot  f_t(\cdot,0)   \;\;\; \mbox{[as $f_t(\cdot,1)+ g_t(\cdot,1)=1$ and
   $f_t(\cdot,1) \cdot g_t(\cdot,1) = 0$]} \\
& = & f_t(\alpha_1,\cdots,0) = g_t(\alpha_1,\cdots,0)\\
& & \mbox{[$h_t()$ sends $\alpha_t$ to the value where $f_t() = g_t()$]} \\[24pt]
\bullet \; f_t(\alpha_1,\cdots,1) \;\; &=& g_t(\alpha_1,\cdots,1)  \;\; \mbox{ but } \;\;
f_t(\alpha_1,\cdots,0) \neq g_t(\alpha_1,\cdots,0)  \\ 
&\Downarrow & \\
h_t(\alpha_1,\cdots,\alpha_t) & = & (\alpha_t +1) \cdot  f_t(\cdot,1) +
 \alpha_t \cdot  f_t(\cdot,1)   \;\;\; \mbox{[as $f_t(\cdot,0)+
   g_t(\cdot,0)=1\; ; \;
   f_t(\cdot,0) \cdot g_t(\cdot,0) = 0$]} \\
& = & f_t(\alpha_1,\cdots,1) = g_t(\alpha_1,\cdots,1) \\
& & \mbox{[$h_t()$ sends $\alpha_t$ to the value where $f_t() = g_t()$]} \\[24pt]
\end{eqnarray*}
\begin{eqnarray*}
\bullet \; f_t(\alpha_1,\cdots,1) \;\; & \neq &h_t(\alpha_1,\cdots,1)  \;\; \mbox{ and } \;\;
f_t(\alpha_1,\cdots,0) \neq g_t(\alpha_1,\cdots,0)  \\ 
&\Downarrow & \;\;
\mbox{[Impossibility to merge the two functions $\; f_t() \;$ and $\; g_t()
    \;$]} \\
&    \Downarrow & \;\;
\mbox{[No constraint over $\alpha_t$ but a induced constraint over some
  $\alpha_j, \; j<t$]} \\
h_t(\alpha_1,\cdots,\alpha_t) & = &  \alpha_t    \;\;\; \mbox{[as $f_t()+ g_t()=1$ and
   $f_t() \cdot g_t() = 0$]} \\
&\mbox{\bf but} & \;\; 
[f_t(\cdot,0) +  
 g_t(\cdot,0)]  \; \cdot \;   [f_t(\cdot,1) +  g_t(\cdot,1)] = fonction(\alpha_1,\cdots, \alpha_j) = 1 \\
\mbox{\bf and} \;\;\;
  g^*_j(\alpha_1,\cdots, \alpha_j)& :=  & 
  [f_t(\cdot,0) +  
 g_t(\cdot,0)]  \; \cdot \;   [f_t(\cdot,1) +  
 g_t(\cdot,1)] +  g_j(\alpha_1,\cdots, \alpha_j) = 1\\
& &  \; \; \; \mbox{[New additional constraint over $\alpha_j, \;\; j <
      t$, so that the impossibility } \\
& & \; \; \; \; \mbox{cannot occur anymore,
  as  $g_j(\cdot, \alpha_j) \rightarrow  g_j(\cdot, \alpha_j) +1$ when it appears.]}
\end{eqnarray*}
\end{proof}
\noindent {\it Example : } \\
\noindent Consider the following sets of clauses : \\
$\bullet \; \; \varphi = (x_1 \vee x_2 \vee \neg x_3) \wedge (\neg x_2 \vee  x_3
\vee \neg x_4) \wedge (\neg x_1 \vee  x_3
\vee \neg x_4)$  \\ 
$\bullet \;\; \varphi' = (\neg x_1 \vee x_2 \vee \neg x_3) \wedge (\neg x_2 \vee  x_3
\vee \neg x_5) \wedge (\neg x_1 \vee  x_3
\vee \neg x_5)$ \\
\noindent Then 
\begin{eqnarray*} [\varphi] & = &  \left[ \begin{array}{c}
\alpha_1\\
\alpha_2\\
\alpha_1 \alpha_3 + \alpha_2 \alpha_3 + \alpha_1 \alpha_2 \alpha_3\\
\alpha_4 \\
\alpha_5
\end{array}
\right]
\wedge 
\left[ \begin{array}{c}
\alpha_1\\
\alpha_2\\
\alpha_3 \\
\alpha_4 + \alpha_2 \alpha_4 + \alpha_2 \alpha_3 \alpha_4\\
\alpha_5
\end{array}
\right]
\wedge 
\left[ \begin{array}{c}
\alpha_1\\
\alpha_2\\
\alpha_3 \\
\alpha_4 + \alpha_1 \alpha_4 + \alpha_1 \alpha_3 \alpha_4\\
\alpha_5
\end{array}
\right] \\
& = & \left[ \begin{array}{c}
\alpha_1\\
\alpha_2\\
\alpha_1 \alpha_3 + \alpha_2 \alpha_3 + \alpha_1 \alpha_2 \alpha_3 \\
\alpha_4 + \alpha_1 \alpha_4 + \alpha_2 \alpha_4 +  \alpha_1 \alpha_2
\alpha_4  + \alpha_1 \alpha_3 \alpha_4 + \alpha_2 \alpha_3 \alpha_4 +
\alpha_1 \alpha_2 \alpha_3 \alpha_4\\
\alpha_5
\end{array}
\right] 
\end{eqnarray*}
\begin{eqnarray*} [\varphi'] & = &  \left[ \begin{array}{c}
\alpha_1\\
\alpha_2\\
\alpha_3 + \alpha_1 \alpha_3 + \alpha_1 \alpha_2 \alpha_3\\
\alpha_4 \\
\alpha_5
\end{array}
\right]
\wedge 
\left[ \begin{array}{c}
\alpha_1\\
\alpha_2\\
\alpha_3 \\
\alpha_4 \\
\alpha_5 + \alpha_2 \alpha_5 + \alpha_2 \alpha_3 \alpha_5
\end{array}
\right]
\wedge 
\left[ \begin{array}{c}
\alpha_1\\
\alpha_2\\
\alpha_3 \\
\alpha_4 \\
\alpha_5 + \alpha_1 \alpha_5 + \alpha_1 \alpha_3 \alpha_5
\end{array}
\right] \\
& = & \left[ \begin{array}{c}
\alpha_1\\
\alpha_2\\
\alpha_3 + \alpha_1 \alpha_3 + \alpha_1 \alpha_2 \alpha_3 \\
\alpha_4 \\
\alpha_5 + \alpha_1 \alpha_5 + \alpha_2 \alpha_5 +  \alpha_1 \alpha_2
\alpha_5  + \alpha_2 \alpha_3 \alpha_5
\end{array}
\right] 
\end{eqnarray*}
\noindent And \\
\begin{eqnarray*}
[\varphi] \wedge [\varphi'] = \left[ \begin{array}{c}
\alpha_1\\
\alpha_2\\
\alpha_2 \alpha_3 \\
\alpha_4 + \alpha_1 \alpha_4 + \alpha_2 \alpha_4 +  \alpha_1 \alpha_2
\alpha_4  + \alpha_1 \alpha_3 \alpha_4 + \alpha_2 \alpha_3 \alpha_4 +
\alpha_1 \alpha_2 \alpha_3 \alpha_4\\
\alpha_5 + \alpha_1 \alpha_5 + \alpha_2 \alpha_5 +  \alpha_1 \alpha_2
\alpha_5  + \alpha_2 \alpha_3 \alpha_5
\end{array}
\right] \\
\end{eqnarray*}
\noindent In this example, no recursive call is done. We get $h_3(\cdot) :=
(\alpha_1 \alpha_3 + \alpha_2 \alpha_3 + \alpha_1 \alpha_2 \alpha_3) \wedge
(\alpha_3 + \alpha_1 \alpha_3 + \alpha_1 \alpha_2 \alpha_3) \;
\stackrel{\mbox{(\ref{merge})}}{\Rightarrow} h_3(\cdot) := (\alpha_3 + 1) \cdot \{0 \cdot \alpha_2 + 0
\} + \alpha_3 \cdot \{ (\alpha_2 + 1) \cdot 0 + (\alpha_2 + 1) \cdot 0
+ \alpha_2 \} = \alpha_2 \alpha_3$, as $f_3(\alpha_1,\alpha_2,0) = 0$,
$f_3(\alpha_1,\alpha_2,1) = \alpha_1 + \alpha_2 + \alpha_1 \alpha_2$, 
$g_3(\alpha_1,\alpha_2,0) = 0$,
$g_3(\alpha_1,\alpha_2,1) = 1 + \alpha_1 + \alpha_1 \alpha_2$,
$f_3(\alpha_1,\alpha_2,1) + g_3(\alpha_1,\alpha_2,1) = \alpha_2 + 1$ 
and $f_3(\alpha_1,\alpha_2,1) \cdot g_3(\alpha_1,\alpha_2,1) =
\alpha_2$. \\ 
\section{An approach to the 3-CNF-SAT problem via descriptors}
\subsection{Boolean descriptors}
\noindent The usual presentation of a 3-CNF-SAT problem consists of a
list on $m$ 3-CNF clauses defined over $n$ propositional variables.
These $m$ clauses describe perfectly the set of solutions for the
3-CNF-SAT problem and can be considered as {\bf  Boolean
  descriptors} of the 3-CNF-SAT problem.  These Boolean descriptors 
are of {\bf linear complexity} as
they can be represented by an array of dimension $3 \times m$.\\

\noindent  The difficulty with such Boolean descriptors is that there
is {\bf no simple or direct relation} between them and the set of solutions or the answer to the
satisfiability question. 
\subsection{3-CNF-matrix descriptors}
\noindent This paper proposes in (\ref{no3}) a 3-CNF-matrix
description of a 3-CNF-SAT problem.  These descriptors (each line in
the 3-CNF-matrix) can be of {\bf exponential complexity} as there are as
many descriptors as solutions. Even if one uses the reduction version of the 3-CNF-matrix
description as explained in (\ref{no5}), simulations show that the complexity remains
exponential.

\noindent The interest of these descriptors is the {\bf
  direct link} between them and the set of solutions or the answer to
the satisfiability question.
\subsection{Functional descriptors}
\noindent This paper proposes also in (\ref{h-def}) a 3-CNF-matrix {\bf
  functional description} for any 3-CNF-SAT problem.  These functional
descriptors are of {\bf unknown complexity}, at least at this stage of
the paper.

\noindent These functional descriptors are {\it somehow in between} both
previous types of descriptors, as they are in an {\bf exponential relation}
 to the set of solutions and in an {\bf direct relation}
 with the satisfiability question.  Indeed, given the
functional descriptors, it is straightforward to give the answer to
the satisfiability question : no if the functional descriptors does
not exist, and yes otherwise.  However, one needs to generate all 
possible values for  $\alpha_i$ to get the entire set of solutions,
and that can take an exponential time. \\

\noindent {\it {\bf Conclusion} : The approach of the 3-CNF-SAT problem via 
  functional descriptors seems to be promising as it does not
  consider the set of all solutions, but only focuses on the sole
  question about satisfiability.}
\pagebreak
\section{Complexity analysis of the functional descriptor approach}
\subsection{A first measure of the complexity for functional
  descriptors}
\begin{theorem} {\bf The complexity of the functional descriptor
    approach for a 3-CNF-SAT problem with $m$ clauses and $n$
    propositional variables is} 
\[{\cal
  O}( m \; n^2 \; \max_{1 \leq t \leq n} \; \max_{1 \leq j \leq m} \len_j(h_t) \; ) \]
where $\len_j(h_t)$ is the number of terms in $h_t(\cdot)$ when
the  $j$ first
  clauses are considered : 
\begin{eqnarray}
\mbox{Let } \; \len(h_t) &\equiv& \sum_{(\delta_1,\cdots,\delta_t)\in
  \{0,1\}^t} \Delta_t(\delta_1,\cdots,\delta_t) \;\; \mbox{[ see
  (\ref{delta-def}) for the definition of $\Delta_t$ ]} \label{len} \\
\mbox{So } \; \len_j(h_t) &=& \len(h_t) \;\; \mbox{ at stage
  $j$ of the computations.} \nonumber 
\end{eqnarray}
\end{theorem}
\mbox{}\\
\begin{proof}
Let us compute the complexity of $f_t(\cdot) \wedge g_t(\cdot)$ in
(\ref{merge}).  First of all, one has to compute the four functions in square
brackets :
$[f_t(\cdot,0)+g_t(\cdot,0)]$, $[f_t(\cdot,1)+g_t(\cdot,1)]$, $[f_t(\cdot,0) \cdot
g_t(\cdot,0)]$ and $[f_t(\cdot,1) \cdot g_t(\cdot,1)$].  We have :
\begin{eqnarray*}
\len(f_t(\cdot,0)) \leq \len(f_t(\cdot,\alpha_t)) \;\; &\mbox{and}& \;\;
\len(f_t(\cdot,1)) \leq \len(f_t(\cdot,\alpha_t)) \; [\equiv \len(f_t)] \\
\len(g_t(\cdot,0)) \leq \len(g_t(\cdot,\alpha_t)) \;\; &\mbox{and}& \;\; 
\len(g_t(\cdot,1)) \leq \len(g_t(\cdot,\alpha_t)) \; [\equiv \len(g_t)] \\
\len(f_t + g_t)  \leq \len(f_t) +
\len(g_t) &\leq& 
\len(f_t) \cdot \len(g_t) \;\;\;\; \mbox{when} \;\;
\len(f_t) > 2 \mbox{ and } \len(g_t) > 2
\end{eqnarray*} 
The complexity for the four functions is then ${\cal O}(
\len(f_t) \cdot \len(g_t ))$. \\
\noindent The complexity for computing $h_t(\cdot)$ in (\ref{merge})
is :
\begin{eqnarray}
{\cal O}(&3 \; \cdot& [( \len(f_t) \cdot \len(g_t) )^2 + ( \len(f_t) \cdot
\len(g_t) ) ] + 2 \cdot [ 2 ( \len(f_t) \cdot \len(g_t) )^2 + ( \len(f_t) \cdot
\len(g_t) )] ) \nonumber \\
&=&  {\cal O}( 7 \cdot ( \len(f_t) \cdot \len(g_t) )^2 + 5 \cdot ( \len(f_t) \cdot
\len(g_t) ) ) \nonumber \\
&=& {\cal O}( [\len(f_t) \cdot \len(g_t)]^2 ) \;\;\; \mbox{ for large }
\len(f_t) \cdot \len(g_t)   \label{complexity}
\end{eqnarray}
\noindent Note : it needs three
runs over the formula in the brackets to do the product with $(\alpha_t +1)$,
one to compute the formula, one to multiply it by $\alpha_t$ and one to
add both results. Similarly, it takes two runs to compute the product with $\alpha_t$.\\[12pt]
\noindent Using the same argumentation, we have :
\begin{eqnarray}
\len(h_t) &=& {\cal O}( [\len(f_t) \cdot \len(g_t)]^2 ) \;\;\; \mbox{ for large }
\len(f_t) \cdot \len(g_t)   \label{len_h} \\
& &\mbox{ and  for the recursive call with } j < t  \;\;\; \mbox{[see
  (\ref{recursive})]} \nonumber \\
 \len(g_j^*) &=& {\cal O}( [\len(f_t) \cdot \len(g_t)]^2 ) \;\;\; \mbox{ for large }
\len(f_t) \cdot \len(g_t)   \label{len_g}
\end{eqnarray}
\noindent To solve the 3-CNF-SAT problem, one should compute all $n$
functional descriptors $h_t(\cdot)$, each of them with at most $n$
recursive calls, and this for each step of integration of
the $m$ clauses.  So, using the equivalence between (\ref{len_h}) and (\ref{complexity}), the overall complexity of the functional
approach to 3-CNF-SAT problem will be of order 
{$\displaystyle {\cal
  O}( m \; n^2 \;\max_{1 \leq t \leq n} \; \max_{1 \leq j \leq m} \len_j(h_t) \; )$}.
\end{proof}
\subsection{Non uniformly distributed versus uniformly distributed
  literals in 3-CNF-SAT problems}
\begin{theorem}
{\bf The most difficult 3-CNF-SAT problems are uniformly distributed ones.}
\end{theorem}
\noindent {\it Note} : The invariance structure of
3-CNF-SAT problem is important with respect to the complexity of the
functional descriptor approach.  It is then normal that problems with uniformly
distributed propositional variables are harder as no re-labeling of
the variables can be done to reduce the complexity.  A simple example
of the importance of re-labeling
is proposed just  after the following proof.  \\[12pt]
\begin{proof}
\mbox{}\\
$\bullet$ {\bf Negative and positive literals} \\ 
First, let us note that the computations involving negative
literals are easier and faster than for positive
ones.  This is a mere consequence of our definition for
$h_t(\cdot)$ in (\ref{def_h}). So {\it the most difficult problems will be
the balanced one with respect to the proportion of negative and
positive literals}. Otherwise, we inverse some variables in
order to get the maximum of negative literals. Let us suppose from here that
the proportion of positive and negative literals is quasi equal for each variable.   \\[12pt]
$\bullet$ {\bf Some definitions} \\
Let us divide now the clauses in two sets. The first set
contains all the clauses with the higher indexed literal being
positive and the second with the negative ones~:
\begin{eqnarray*} 
Cl^+ &=& \bigcup_t \;\; Cl^+(x_t) \\
&=& \bigcup_t \; \{ \psi_i := [\neg] x_r \vee [\neg] x_s \vee
x_t \; \mbox{  with  } \; r < s < t, \; \mbox{ or any permutation of  } \; x_r,
x_s, x_t \}  \\
Cl^- &=& \bigcup_t \;\; Cl^-(x_t) \\
 &=& \bigcup_t \; \{ \psi_i := [\neg] x_r \vee [\neg] x_s \vee
\neg x_t \; \mbox{  with  } \; r < s < t, \; \mbox{ or any permutation of  } \; x_r,
x_s, x_t \} 
\end{eqnarray*}
\begin{eqnarray*}
\mbox{ and } \; V^+(x_t) &=& \{x_i  \; (i < t)  \; | \; \exists  \; \psi  \; \in Cl^+(x_t) : x_i \;
\mbox{ appears in } \; \psi\}  \\  
V^-(x_t) &=&  \{x_i  \; (i < t)  \; | \; \exists  \; \psi  \; \in Cl^-(x_t) : x_i \;
\mbox{ appears in } \; \psi\} \\
V(x_t) &=& V^+(x_t) \cup V^-(x_t) 
\end{eqnarray*}
By construction, there is at least one solution for each $x_t$ when
considering clauses only in $Cl^+ [x_t = 1]$ or in $Cl^- [x_t=0]$.
Moreover, the computation of the functional descriptors will not
involved recursive calls [see (\ref{recursive})] as no impossibility exists
for any $x_t$.\\[12pt]
As $h_t(\cdot)$ is a multi-linear
combination of the $\alpha_i$ corresponding to the literals $x_i$ ($i
\leq t$) found in clauses where $x_t$ has the highest index, there will
be at most $2^{(\#V^+(x_t) \;+\;1)}$ [respectively $2^{(\#V^-(x_t) \;+\;1)}$] terms in
$h_t(\cdot)$ computed over $Cl^+$ [resp. $Cl^-$]. So
$\len(h_t)_{|Cl^+} \leq 2^{(\#V^+(x_t) \;+\;1)}$ and
$\len(h_t)_{|Cl^-} \leq 2^{(\#V^-(x_t) \;+\;1)}$. \\[12pt]
$\bullet$ {\bf The merging of $Cl^+$ and $Cl^-$}. \\
 One needs to compute $h_t (\cdot)  \;  :=  \;  h_t(\cdot)_{|Cl^+}
 \wedge h_t(\cdot) _{|Cl^-}$. \\
$\diamond$ For $t=n$ and considering that $h_n(\cdot)$ is a multi-linear combination of the
$\alpha_i$ appearing in $h_n(\cdot)_{|Cl^+}$ or $h_n(\cdot)
_{|Cl^-}$, we have that :
\begin{eqnarray} \len(h_n) &\leq& 2^{\#(V^+(x_n) \; \cup \; V^-(x_n))
  \;+\;1} \nonumber \\
&\leq& 2^{\#V(x_n)   \;+\;1} \label{merging_complexity} 
\end{eqnarray}
$\diamond$ For $t=n-1$ also, $h_{n-1}(\cdot)$ is a multi-linear
combination of the $\alpha_i $ corresponding to the variables $x_i \; \; (i < n-1)$
found in common clauses with $x_{n-1}$ as the highest variable~: $\len(h_{n-1}) \leq 2^{\#V(x_{n-1})  \;+\;1}$.  
 {\bf But} one has {\bf to add} the potential $\alpha_i$
involved in the recursive call $g^*_{n-1}(\cdot)$  from the previous computation of
$h_n(\cdot)$. [see
(\ref{recursive})] \\[12pt]
\mbox{  }$\;\;$ From (\ref{recursive}), $g^*_{n-1}(\alpha_1, \cdots, \alpha_{n-1}) =
[h_n(\cdot,0)_{|Cl^+} + h_n(\cdot,0)_{|Cl^-}] \cdot
[h_n(\cdot,1)_{|Cl^+} + h_n(\cdot,1)_{|Cl^-}]$.  So $g^*_{n-1}(\cdot)$ will be a
multi-linear combination of the same $\alpha_i$, except $\alpha_n$, as
for $h_n(\cdot)$.  Therefore, $h_{n-1}(\cdot) \wedge g^*_{n-1}(\cdot)$
will be a combination of the $\alpha_i$ associated to the variables in
$ \bigcup_{i=n-1}^{n} \; [V^+(x_{i}) \; \cup \; V^-(x_{i})]
\Rightarrow  \len(h_{n-1} \wedge g^*_{n-1}) \leq
2^{\#[\bigcup_{i=n-1}^{n}\; V(x_{i}) ]}$, as $x_{n-1}$ should not be
counted in $V(x_n)$. \\[12pt]
$\diamond$ So, $\forall \; t \; : \; \len(h_{t} \wedge g^*_{t}) \leq
2^{\#[\bigcup_{i=t}^{n}\; V(x_{i}) ]\;+\; 1-(n-t)}$.  But as $h_{t} \wedge g^*_{t}$ is a combination of at most
$t$ $\alpha_i$'s, we have :
\begin{eqnarray}
 \len(h_{t} \wedge g^*_{t}) \leq
\min (\; 2^{\#[\bigcup_{i=t}^{n}\; V(x_{i})]
  \;+\;1-(n-t)} \; , \; 2^t \;) \label{len_h_g} 
\end{eqnarray}
$\bullet$ {\bf Uniform distribution of the literals} \\
As $V(x_{i})$ is dependent of the ordering of
the variables, it is possible to re-order the variables so that 
$2^{\#[\bigcup_{i=t}^{n}\; V(x_{i})]}$ is
  minimal, except in the case of uniformly distributed literals. {\bf
    The uniformly distributed case is then the most difficult problem},
    as no ordering can reduce the maximum value in (\ref{len_h_g}). \\
\end{proof}
\mbox{}\\
\noindent {\it Example of non uniformly distributed 3-CNF-SAT problem  : } \\ 
$\bullet$ Consider the following 3-CNF-SAT problem with $m$ clauses and $2 m + 1$ propositional variables : 
\begin{eqnarray*}  
\varphi := \bigwedge_{1 \leq i \leq m} (x_{2i-1} \vee x_{2i} \vee x_{2
  m + 1}) 
\end{eqnarray*}
Here we have : 
\begin{eqnarray*}
V^+(x_t) &=& V^-(x_t) =  \emptyset \; \; \forall \; t \neq 2 m + 1 \\
V^+(x_{2  m + 1}) &=& \{x_1, \cdots, x_{2 m} \} \;\; \mbox{ and } \;\;
V^-(x_{2 m + 1}) = \emptyset \\
\max_{1 \leq t \leq n} \;
    \max_{1 \leq j \leq m} \len_j(h_t) \; &=& 3^{m+1} - 3^m + 1 = {\cal
      O}(3^m) 
\end{eqnarray*}
Note : The proof of this equality is more difficult than
  interesting, so we do not write it here. \\[12pt]
$\bullet$ But the {\it same} 3-CNF-SAT problem can be formalized in terms of opposite
literals \linebreak $y_i = \neg x_i \;,\; \forall i \in \{1,
\cdots, 2 m +1\} $ :
\begin{eqnarray*}  
\varphi := \bigwedge_{1 \leq i \leq m} (\neg y_{2i-1} \vee \neg y_{2i}
\vee \neg y_{2
  m + 1}) 
\end{eqnarray*}
This time, we have : 
\begin{eqnarray*}
V^+(y_t) &=& V^-(y_t) =  \emptyset \; \; \forall \; t \neq 2 m + 1 \\
V^+(y_{2  m + 1}) &=& \emptyset \;\; \mbox{ and } \;\;
V^-(y_{2 m + 1}) = \{y_1, \cdots, y_{2 m} \}  \\
\max_{1 \leq t \leq n} \;
    \max_{1 \leq j \leq m} \len_j(h_t) \; &=& 2^m = {\cal
      O}(2^m) 
\end{eqnarray*}
\begin{eqnarray*}
\mbox{Indeed,} \;\;\; h_t(\alpha_1, \cdots, \alpha_t) &=& \alpha_t \;\;
\forall \; t < 2m+1 \\
\mbox{and} \;\;\; h_{2m+1}(\alpha_1,\cdots,\alpha_{2m+1}) &=&
[(\alpha_1 \alpha_2 \alpha_{2m+1} + \alpha_{2m+1})
\wedge (\alpha_3 \alpha_4 \alpha_{2m+1}  + \alpha_{2m+1})]
\wedge \cdots \\
&\stackrel{see(\ref{merge})}{=}& [(\alpha_{2m+1}+1) \cdot 0 + (\alpha_{2m+1}) \cdot (\alpha_1
\alpha_2 + 1) \cdot (\alpha_3 \alpha_4 + 1)] \wedge \cdots \\
&=& (\alpha_{2m+1}) \cdot \prod_{i=1}^{m} (\alpha_{2i-1} \;
\alpha_{2i} + 1) \\
&\Rightarrow& \len_{2m+1}\;(h_{2m+1}) = 2^m \\
\end{eqnarray*}
$\bullet$ Finally, the {\it same} 3-CNF-SAT problem can be formalized using
re-ordered propositional variables $z_1 = y_{2 m + 1}$ and $z_{i} =
y_{i-1} \; , \; \forall i \in \{2,\cdots,2 m + 1\}$ :
\begin{eqnarray*}  
\varphi := \bigwedge_{1 \leq i \leq m} (\neg z_{2i}
\vee \neg z_{2 i +1} \vee \neg z_1) 
\end{eqnarray*}
And, this time, we have :
\begin{eqnarray*}
V^+(z_t) &=&  \emptyset \; \; \forall \; t \\ 
V^-(z_t) &=& \left\{ \begin{array}{l} 
\emptyset \;\; \mbox{ for } \;\; t = 1 \;\mbox{ or } \; t = 2 i \; \; (1 \leq i \leq m) \\
\{z_1, z_{2i}\}  \; \; \mbox{ for } \;\; t = 2 i + 1 \; \; (1 \leq i
\leq m) 
\end{array}
\right.\\
\max_{1 \leq t \leq n} \;
    \max_{1 \leq j \leq m} \len_j(h_t) \; &=& 2 
\end{eqnarray*}
So for this example, one can reach a {\bf linear complexity}
      of  ${\cal O}(2 \cdot \mbox{number of } h_t(\cdot)) = {\cal
      O}(2 \; m)$,  as only one $h_t(\cdot)$ has to be computed at
  each step without any recursive call. \\[12pt]
{\it Example of uniformly distributed 3-CNF-SAT problem :} \\
The smallest exact uniformly distributed and optimally re-ordered 3-CNF-SAT problem is~:
\begin{eqnarray*}
\varphi := \bigwedge_{i=1}^{8} \psi_i = \bigwedge \left\{ \begin{array}{l}
x_1 \vee \neg x_2 \vee \neg x_3 \\
 x_1 \vee  x_2 \vee \neg x_3 \\
\neg x_1 \vee \neg x_2 \vee \neg x_3 \\
\neg x_1 \vee  x_2 \vee \neg x_3 \\
x_1 \vee \neg x_2 \vee  x_3 \\
x_1 \vee x_2 \vee x_3 \\
\neg x_1 \vee \neg x_2 \vee x_3 \\
\neg x_1 \vee x_2 \vee x_3 
\end{array}
\right.
\end{eqnarray*}
No relabeling will reduced the 3-CNF-SAT complexity.  This problem is
{\it ``hard''} in the sense that  each clause eliminates only one
solution at a time. We
have here  :
\begin{eqnarray*}
V^+(x_1) &=& V^-(x_1) = \emptyset \\
V^+(x_2) &=& V^-(x_2) = \emptyset \\
V^+(x_3) &=& V^-(x_3) = \{x_1,x_2\} \\
&\mbox{and}&
\end{eqnarray*}
\begin{eqnarray*}
\left[ \begin{array}{llllcc}
\mbox{Step \hspace{1cm} } &h_1(\cdot)&h_2(\cdot)&h_3(\cdot) & \max_t
\len(h_t) & \# \mbox{Solutions}\\
\psi_1& \alpha_1 & \alpha_2 &
(\alpha_1 + 1) \alpha_2 \alpha_3 +\alpha_3 & 3 &7 \\
\psi_1 \wedge \psi_2 &\alpha_1 &  \alpha_2 &
\alpha_1 \alpha_3 &  1 & 6\\
\wedge_{i=1}^{3} \psi_i & \alpha_1 & \alpha_2
& \alpha_1 \alpha_2 \alpha_3 + \alpha_1 \alpha_3 & 2 & 5\\ 
\wedge_{i=1}^{4} \psi_i  & \alpha_1 & \alpha_2 &
 0 & 1 & 4 \\
\wedge_{i=1}^{5} \psi_i  & \alpha_1 & \alpha_1 \alpha_2 & 0
& 1 & 3 \\
\wedge_{i=1}^{6} \psi_i  & 1 &  \alpha_2 & 0 & 1 & 2\\
\wedge_{i=1}^{7} \psi_i  & 1 & 0 & 0 & 1 & 1\\
\wedge_{i=1}^{8} \psi_i  & \nexists & \nexists & \nexists & 0 & 0  \\
\end{array}
\right]\\[12pt]
\end{eqnarray*}
{\bf Conclusions :} \\
The theorem about uniformity is important as it states : for any non uniformly
distributed 3-CNF-SAT problem $\varphi$
with $m$ clauses and $n$ variables, there exists an uniformly
distributed 3-CNF-SAT problem $\varphi'$ with $m$ clauses and $n$ variables which
is more difficult to solve, in terms of functional descriptors.\\
\subsection{A sorting algorithm to reduce complexity}
\noindent We propose the following sorting algorithm of complexity
${\cal O}(m + n \log(n) + m \log(m) )$ :
{\small \textsf{ 
\begin{itemize}
\item[$\diamond$] Relabel the propositional variables [$x_i \rightarrow y_j$] in order to get their
occurrence [${\cal O}(m)$] in a decreasing order :  $\#\{y_1\}$ is maximal, $\cdots$,
$\#\{y_n\}$ is minimal [${\cal O}(n\log(n))$];
\item[$\diamond$] Inverse the sign of the literals in order to get the maximum of
negative literals; 
\item[$\diamond$] Sort the clauses to get a increasing order of the highest variable
in the ordered clauses [${\cal O}(m\log(m))$];
\item[$\diamond$] Within the set of clauses with the same highest variable, sort the
clauses so that the ones with negative highest variable appear before the
ones with positive highest variable.\\ 
\end{itemize}
}  }
\noindent As we seldom have exact uniformly
distributed 3-CNF-SAT problems, the complexity can then be reduced
drastically, as shown in Figure 1.
\begin{figure}[htb]
\centering
\includegraphics[width=6cm]{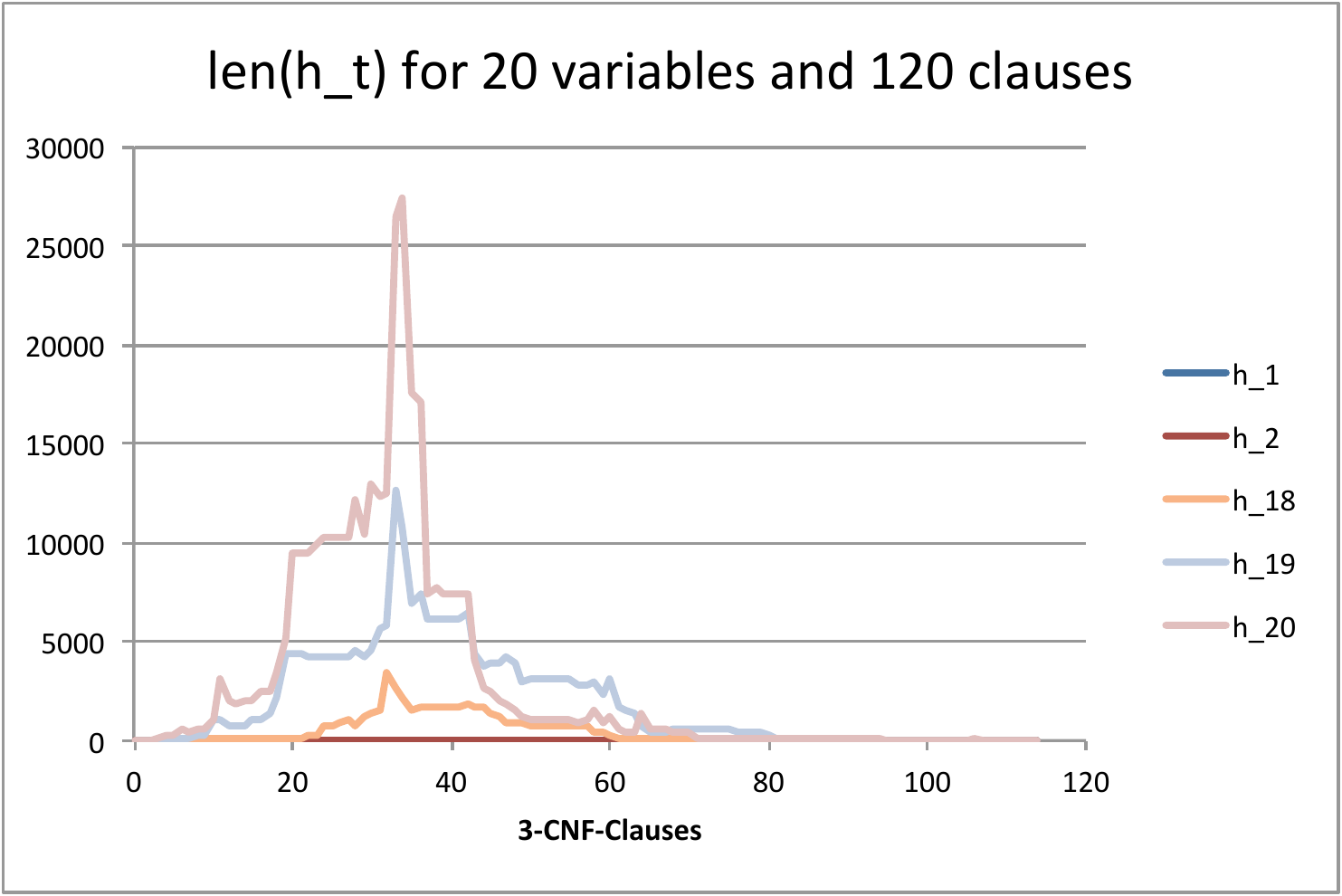} \mbox{$\;\;$}
\includegraphics[width=6cm]{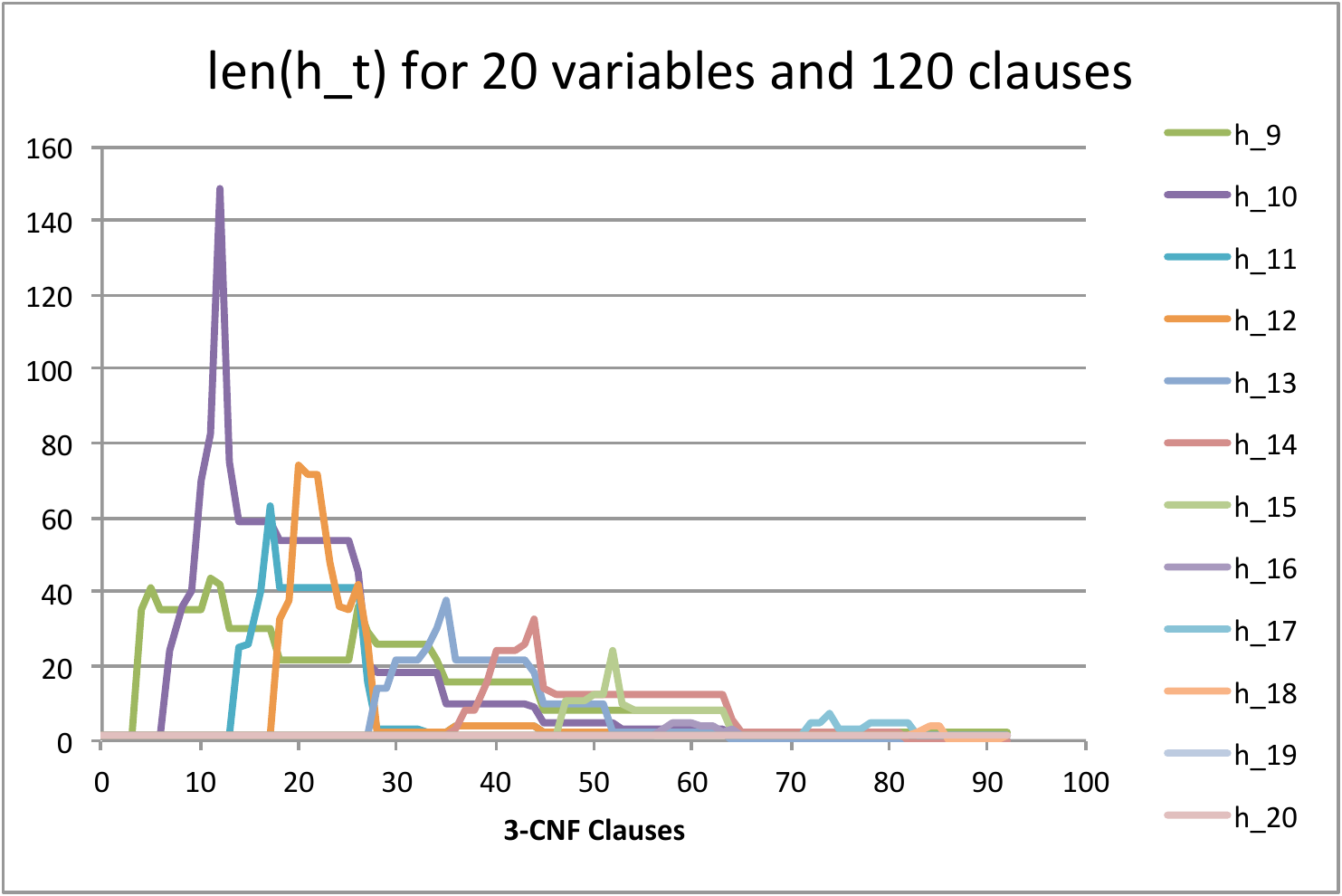} \\ 
{\small Fig. 1 : Complexity for the same dataset before and after the
  sorting algorithm.}
\end{figure}  
\subsection{Exact uniformly distributed $\alpha$-random 3-CNF-SAT problems}
%
{\it Definition} : Let $\varphi$ be a 3-CNF-SAT problem with $n$ variables,
each of them appearing exactly $\frac{3 \alpha}{2} $ times as positive
and $\frac{3 \alpha}{2} $ times as negative literal, for some $\alpha > 0$. Let these $3 \; \alpha \; n$ literals be randomly distributed
amongst the clauses. Such problem is called {\bf \em an  exact uniformly
  distributed $\alpha$-random 3-CNF-SAT problem.} \label{exact_problem} \\[12pt]
\noindent {\bf Remark : } For exact uniformly distributed 3-CNF-SAT problem,
the labeling part of the previous {\it sorting
algorithm} has no effect, as the occurrence of each literal is $\frac{3
  \alpha}{2}$.  Only the re-ordering of the $m$ clauses can reduce
the complexity of the problem.\\[12pt]
\begin{theorem}  
{\bf For such exact uniformly distributed $\alpha$-random 3-CNF-SAT problems,
 the expected number of 
clauses where $i \; (i > 2)$  is the highest index, is noted $m_\alpha(i)$ and given by :}
\begin{eqnarray}
E[\#\{ \psi = [\neg]x_r \vee [\neg]x_s \vee [\neg]x_t |
\max(r,s,t)=i \}] = \frac{(i-1) (i-2)}{(n-1) (n-2)} \; 3 \; \alpha
\equiv m_\alpha(i) \label{m_alpha}
\end{eqnarray}
\end{theorem}
\mbox{}\\
\begin{proof} 
Let $\psi$ be a clause with $x_i$ or $\neg x_i$, there are $C_2^{i-1} \cdot \; 3 \; \alpha$ combinations with smaller indices amongst $C_2^{n-1} \cdot \; 3 \; \alpha$
possibles combinations.  So, the probability for $x_i$ to get the
highest index is : $\frac{(i-1) (i-2)}{(n-1) (n-2)}$. The expected
value is obtained by multiplying the probability by the number of
occurrences of $x_i$.\\
\end{proof}
\noindent Figure 2 shows the theoretical density and cumulative distributions of $m_\alpha(i)$
versus the distributions for the observed values for  $\#\{ \psi = [\neg]x_r \vee [\neg]x_s \vee [\neg]x_t |
\max(r,s,t)=i \}$  in the case of a 3-CNF-SAT problem with $175$
variables and $753$ random clauses.  \\
%
\begin{figure}[htb]
\centering
\includegraphics[width=6cm]{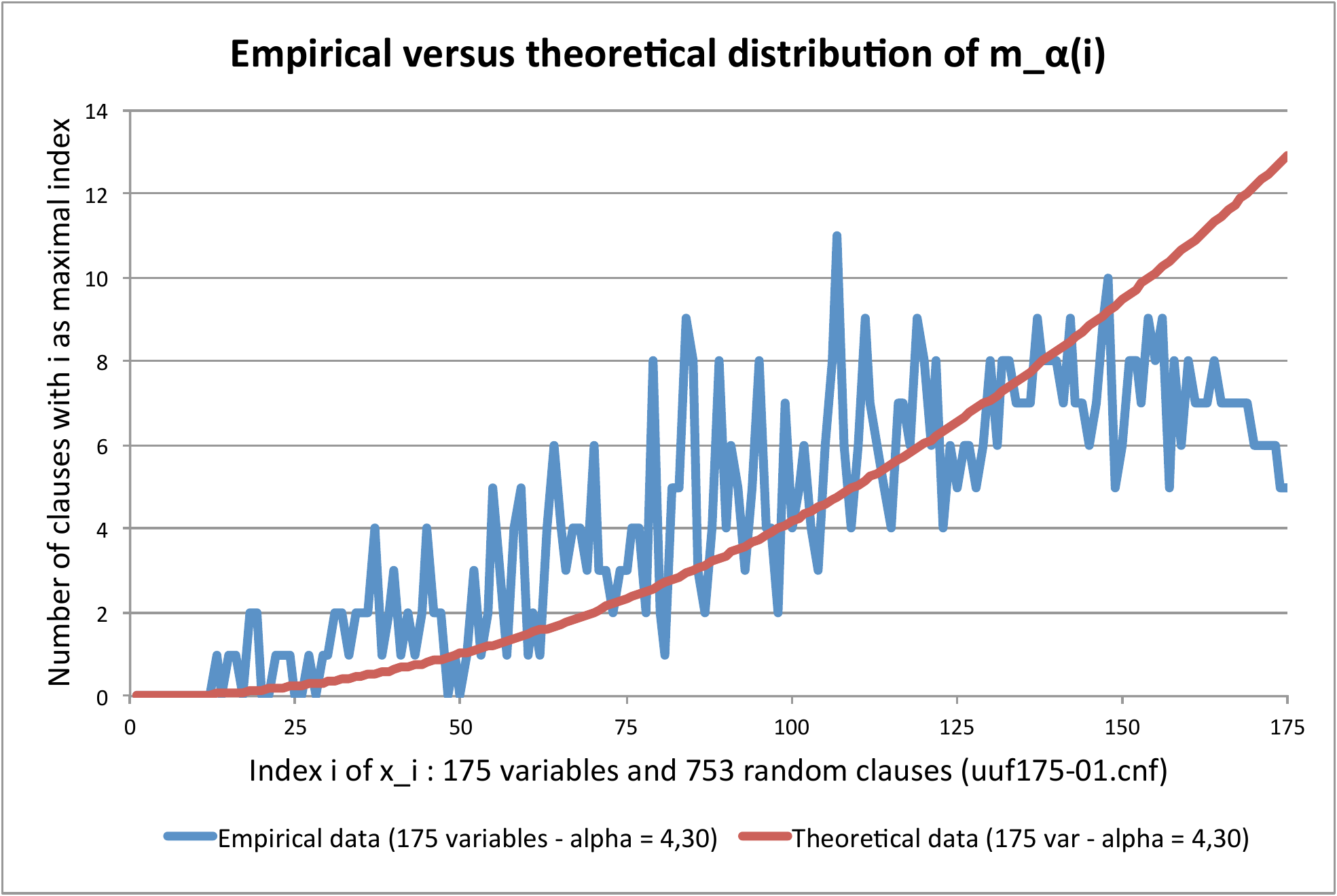}  \mbox{$\;\;$}
\includegraphics[width=6cm]{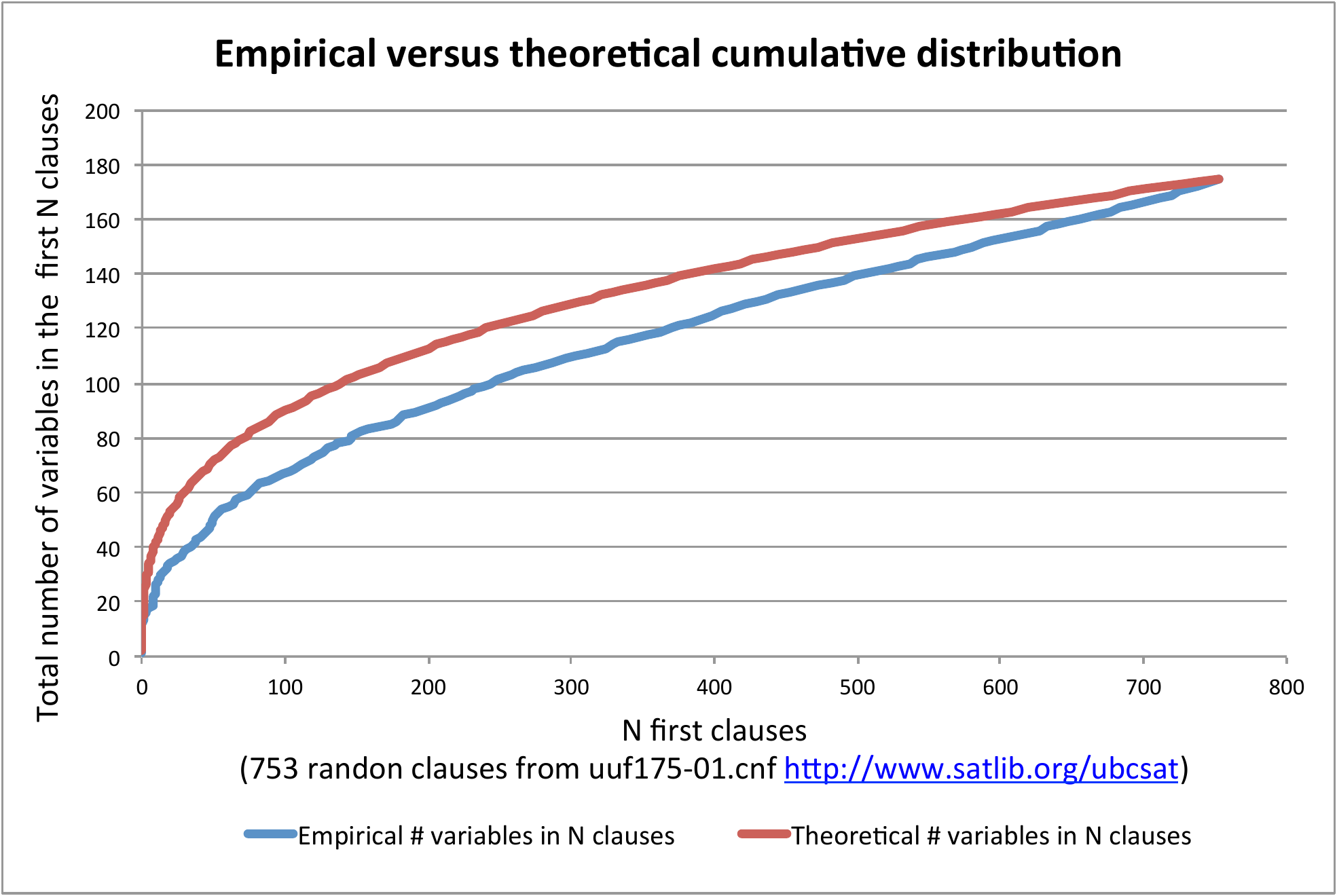} \\
{\small Fig. 2 : Density and cumulative distributions of ``sorted clauses'' for
  $n=175$ and $\alpha = 4,30$.}
\end{figure}  
\begin{theorem} 
{\bf For exact uniformly distributed $\alpha$-random 3-CNF-SAT problems,
 the expected number of variables in $V(x_i)$, for $i > 2$ and
 large $n$, is given by :}
\begin{eqnarray}
 E[\# V(x_i)] = 2 \; m_\alpha(i) = \frac{(i-1)
   (i-2)}{(n-1) (n-2)} \; \; 6 \; \alpha \label{m_i}
\end{eqnarray}
\end{theorem}
\mbox{}\\
\begin{proof}
There is $C_2^{i-1}$ possible triplets with
$x_i$ being the highest indexed variable.  The probability for some
$x_j \; (j < i)$ to appear in one of these triplets is
$\frac{i-2}{C_2^{i-1}} = \frac{2}{i-1} = p $ for any $j$.  The occurrence of
$x_j$ follows a binomial model $Bi(m_\alpha(i),p)$, as one can choose several
times the same triplet (given different clauses with respect to the negative or
positive sign of the included literals). The expected number of occurrence of $x_j$ in the $m_\alpha(i)$
triplets is then $m_\alpha(i) \cdot p = \frac{6 \; \alpha (i-2)}{(n-1) (n-2)} \; <
\; 1$ for large $n$.  So each variable is expected to appear at most
once in the $m_\alpha(i)$ triplets-clauses.  Therefore, the number of
variables, different from $x_i$,
occurring in these $m_\alpha(i)$ clauses is $2 m_\alpha(i)$ as there are two
variables distinct from $x_i$ in each clause.
\end{proof}
\mbox{}\\
\begin{theorem} 
{\bf For exact uniformly distributed $\alpha$-random 3-CNF-SAT problems, the maximal expected 
complexity for the computation of $h_t(\cdot)$ is bounded by :}
\begin{eqnarray*}
\begin{array}{r}
\mbox{}\\[15pt]
\len(h_t) \leq \displaystyle \max_{k \geq 0} 
\end{array} \,
\begin{array}{l}
\mbox{} \;\; \left[
\begin{array}{r} 
\min \{ 2(\displaystyle \sum_{j=0}^k m_\alpha(n^{(j)})) - (k-1)
\; , \; n^{(k)} \} 
\end{array} \right] \\[3pt]
2\\[-10pt]
\mbox{}
\end{array} 
\end{eqnarray*}
{\bf where $ 2(\displaystyle \sum_{j=0}^k m_\alpha(n^{(j)})) - (k-1) \;$
is a concave quadratic function with respect to $t$ or $k$, as shown
on figures 3 and 4.} 
\end{theorem}
\mbox{}\\
\begin{proof}
From (\ref{m_i}), we know that $\#V(x_i)$ is
expected to be maximal for $i=n$, when $\#V(x_n) = 
2 m_\alpha(n) = 6 \; \alpha$.  So, from (\ref{merging_complexity}), we
have that :
\begin{eqnarray}
\len(h_n) & \leq & 2^{\# V(x_n) 
  \;+\;1} \nonumber \\
& \leq & 2^{2 \; m_\alpha(n) + 1} = 2^{(6 \alpha + 1)}
\end{eqnarray}
For the computation of the recursive call $g^*_{j}(\alpha_1,
\cdots,\alpha_j)$ (see \ref{recursive}), the index $j$ is the highest
index of the variables in $V(x_n)$. Let us note it
$n^{(1)}$. We have thus $\# \; V(x_n) = 2 \; m_\alpha(n)$ indexes
uniformly chosen from $\{1, \cdots, n-1\}$. $n^{(1)}$ will be the expected
maximal index from an uniform distribution for $2 \; m_\alpha(n)$
$iid$ variables $u_i \sim  U[1,\cdots,n-1]$ :
\begin{eqnarray}
n^{(1)} & = & E[\;\; \max_{1\leq i \leq 2m_\alpha(n)} \;\; (u_i) ]=
\frac{2\; m_\alpha(n)}{2\; m_\alpha(n) +1} \; (n-1) = \frac{6 \alpha}{6 \alpha +
  1} \; (n-1)  
\end{eqnarray}
So, for the recursive call, we will have to compute $h_{n^{(1)}}(\cdot) \wedge
g^*_{n^{(1)}}(\cdot)$.  We get :
\begin{eqnarray*}
 \# V(x_{n^{(1)}}) & = & 2 \; m_\alpha(n^{(1)}) =
\frac{(n^{(1)}-1)(n^{(1)}-2)}{(n-1)(n-2)} \;6 \; \alpha \;\;\;\;\;\; \mbox{from (\ref{m_i})}\\
\mbox{}\\
\len(h_{n^{(1)}}(\cdot) \wedge
g^*_{n^{(1)}}(\cdot)) & \leq & 2^{\# \{V(x_n) 
 \; \cup \; V(x_{n^{(1)}})\} \;+\; 1-(2-1)} \;\;\;\;\; \mbox{from (\ref{len_h_g})} \\
& \leq & 2^{2(m_\alpha(n) +  m_\alpha(n^{(1)})) } 
\end{eqnarray*}
And so on, for the next recursive calls.  We get for the recursive
call $k \; (k>1)$ :
\begin{eqnarray}
n &\equiv& n^{(0)} \nonumber \\[12pt]
u_i &\sim& U[1, \cdots, n^{(k-1)}-1] \nonumber \\[12pt]
n^{(k)} = E[\;\; \max_{1\leq i \leq 2m_\alpha(n^{(k-1)})}
\;\; (u_i) ] &=&
\frac{2\; m_\alpha(n^{(k-1)})}{2\; m_\alpha(n^{(k-1)}) +1} \;
(n^{(k-1)}-1) \nonumber \\[12pt]
 \#\{V(x_{n^{(k)}}) \} &=& 2 \; m_\alpha(n^{(k)}) = 
\frac{(n^{(k)}-1)(n^{(k)}-2)}{(n-1)(n-2)} \;6 \; \alpha  \nonumber \\[12pt]
\len(h_{n^{(k)}}(\cdot) \wedge
g^*_{n^{(k)}}(\cdot)) & \leq & 2^{\left[
\begin{array}{r}
\min \{ \# \displaystyle \bigcup_{j=0}^k \; \{V(x_{n^{(j)}})\}
 - (k-1) \; , \; n^{(k)} \} 
\end{array} \right] } \nonumber \\[12pt]
\len(h_{n^{(k)}}(\cdot) \wedge
g^*_{n^{(k)}}(\cdot)) & \leq & 2^{\left[
\begin{array}{r} 
\min \{ 2(\displaystyle \sum_{j=0}^k m_\alpha(n^{(j)})) - (k-1)
\; , \; n^{(k)} \}
\end{array} \right] } \nonumber \\[12pt]
&\leq&2^{\left[
\begin{array}{r} 
\min \{ M_\alpha(n^{(k)}), n^{(k)} \} 
\end{array} \right] } \label{complexity_formula} 
\end{eqnarray}
\begin{figure}[htb]
\centering
\includegraphics[height=6cm]{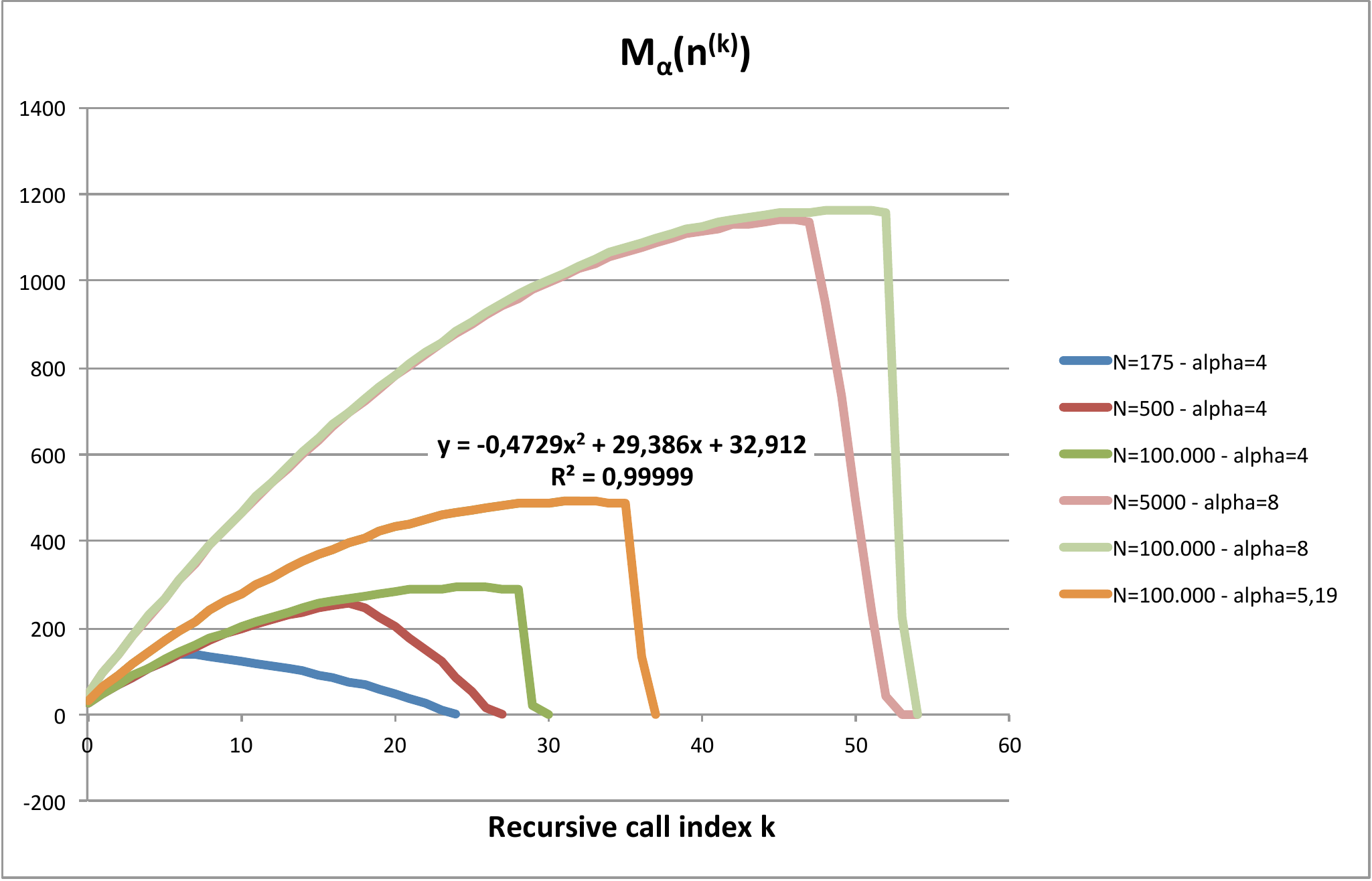} \\
{\small Fig. 3 : Complexity wrt $k$ : $M_\alpha(n^{(k)})
 = 2(\displaystyle \sum_{j=0}^k m_\alpha(n^{(j)})) -
  (k-1) $ 
where $n^{(0)}=n$.}
\end{figure}  
 
\noindent This bound is only defined for the variables with
$n^{(k)}$ as index. Note that $n^{(k)}$ are
functions of the starting index $n^{(0)}=n$.  We can compute similar bounds for other starting
indexes $n^{(0)}$ in $[1, \cdots , n-1]$, so that $M_\alpha(\cdot)$
can be defined for all $t$ as shown in Figure 4.  \\[12pt]
\noindent Numerical computations show that, for large $n$,  $M_\alpha(n^{(k)})$ as well
as $M_\alpha(t)$ are concave
quadratic functions with coefficients only depending on $\alpha$. This
can be easily explained as a mere consequence of the $iid$ randomness of the 
variables  $\#\{V(x_{n^{(k)}}) \}$ and $m_\alpha(n^{(j)})$.  Indeed, the {\it central limit theorem} for
the expectation of 
$iid$ random variables predicts that 
E[$2^{M_\alpha(n^{(k)})}$]  follows a Normal
distribution (censored by $\min$).  But $X \sim N(\mu,\sigma^2)$
implies a quadratic log-density : $\log(f_{X}(x)) \propto
-\frac{(x-\mu)^2}{2 \sigma}$. For each value of $\alpha$, we can
compute the corresponding $\mu_\alpha$, $\sigma_\alpha$ and the
maximum value for $M_\alpha(n^{(k)})$.  Quadratic regression
estimations give $\max_k(M_\alpha(n^{(k)})) 
\approx 294$ for $\alpha=4$, $\max_k(M_\alpha(n^{(k)})) \approx
490$ for $\alpha=5,12$ (see figure 3 and below for the choice of such $\alpha$)  and $\max_t(M_\alpha(t)) \approx
1160$ for $\alpha=8$.
\\[12pt] 
\noindent {\bf Remark : } It is now important to see whether different
starting points $n^{(0)}$
yield not to
aggregating trajectories so that addition of bounds are to be
considered.  This situation can be neglected as shown in the following
theorem.
\end{proof}
\mbox{}\\
\begin{theorem}
{\bf The probability for a given variable $x_i$
to be in more than one trajectory tends to zero for large $n$.}\\
\end{theorem}
\begin{proof}
\noindent Let us consider separately the possible trajectories
$tr(x_{n^{(0)}} \rightarrow x_{n^{(k)}})$ for $n^{(0)}=m \in \{1,
\cdots,n\}$ and $k \in \{1,\cdots,n\}$.  Let us note a given
trajectory : $tr(m,k_m)$ with
$k_m$ such that $n^{(k_m)} > i$. For each
variable $x_i$, there exists at most $(n-i)(n-i-1)/2$ trajectories $tr(m,k_m)$ where $x_i$ could be
the next highest indexed variable for $n^{(k_m+1)}$ :
$tr(n,0), \cdots, tr(n,k_n), \cdots,$\linebreak $tr(i+1,0)$.  The
probability for $x_i$ to get the highest index in a trajectory is :
\begin{eqnarray*}
P[i = \max_l \{l : x_l \in tr(m,k_m)\}] \hspace{-30pt} && \\
&=& P[i = \max_l \{l : x_l \in \hspace{-6pt} \bigcup_{j=0|n^{(0)}=m}^{j=k_m} \hspace{-6pt}
V(x_{n^{(j)}}) \setminus \{x_{n^{(0)}},\cdots,x_{n^{(k_m)}} \} \} ] \\
&=& P[i = \max_l \{l : x_l \in tr(m,k_m) \} | x_i
\in tr(m,k_m)]  \; \cdot \; P[x_i \in tr(m,k_m)] 
\end{eqnarray*}
\noindent We have :
\begin{eqnarray*}
P[x_i \in tr(m,k_m)] &=& \# \{\mbox{clauses in $tr(m,k_m)$}\} \cdot P[x_i
\in \mbox{the clause and } i \mbox{ is the maximum index}] \\
&=& \sum_{j=0|n^{(0)}=m}^{j=k_m} \; m_\alpha(n^{(j)}) \cdot
\frac{n^{(j)}-2}{C_2^{n^{(j)}-1}}   \mbox{\hspace{2cm}
  [see (\ref{m_alpha})]}\\
&=& \sum_{j=0|n^{(0)}=m}^{j=k_m} \; m_\alpha(n^{(j)}) \cdot
\frac{2}{n^{(j)}-1} \\
&=& \sum_{j=0|n^{(0)}=m}^{j=k_m} \; \frac{(n^{(j)}-1)(n^{(j)}-2)}{(n-1)(n-2)(n^{(j)}-1)} \;\; 6 \alpha  \\
&=& \frac{(6 \alpha)}{(n-1)(n-2)} \;\; \sum_{j=0}^{j=k_m} \; (n^{(j)}-2)
\end{eqnarray*}
\noindent For instance : 
\[P[x_i \in tr(m,0)]=\frac{6 \alpha
  (n^{(0)}-2)}{(n-1)(n-2)} = \frac{6 \alpha}{n-1} \frac{m-2}{n-2} \leq \frac{6 \alpha}{n-1}
\mbox{ as } n^{(0)}=m \leq n\]  
and 
\begin{eqnarray*}
P[x_i \in tr(m,1)] &\leq& P[x_i \in tr(n,1)] \\
 &=& \frac{6 \alpha
  }{(n-1)(n-2)} (n^{(0)}-2)+(n^{(1)}-2) \\
 &=& \frac{6 \alpha}{n-1} + \frac{6
    \alpha}{(n-1)(n-2)} ([\frac{6
    \alpha}{6 \alpha+1} (n-1)]-2) \\
&<&  \frac{6 \alpha}{n-1} + \frac{6
    \alpha}{(n-1)(n-2)} (n-3) \\
&<&  \frac{6 \alpha}{n-1} (1+ \frac{n-3}{n-2}) 
\end{eqnarray*}
Finally, 
\begin{eqnarray*}
P[x_i \in tr(m,k_m)]  \leq P[x_i \in tr(n,k_m)] 
&<&  \frac{6 \alpha}{n-1} (1+ \frac{n-3}{n-2} + \cdots + \frac{n-(k_m+2)}{n-2}) \\
&\rightarrow& 0 \;\; \mbox{for large $n$ with respect to $k_m$ and $\alpha$}. 
\end{eqnarray*}
Now, considering that the elements of $tr(m,k_m)$ are $iid$
uniformly distributed random variables drawn from $\{1, \cdots,
m-1\}$, we have : 
\begin{eqnarray*}
P[i = \max_l \{l : x_l \in tr(m,k_m)\} | x_i
\in tr(m,k_m)] \hspace{-40pt} && \\
 &=& P[i=\max \{ \# tr(m,k_m) \mbox{uniform random
  variables}\}] \\
&& \mbox{[for large $n$, we use the
  expected value for $\# tr(m,k_m)$]} \\
&=& P[ \sum_{j=0}^{j=k_m} 2 m_\alpha(n^{(j)}) - 1 \; \; \mbox{ uniform $iid$ variables } \leq i
\;\; ] \\
&=& \prod_{l=1}^{\sum_{j=0}^{j=k_m} 2 m_\alpha(n^{(j)}) - 1} \frac{i}{(m-1)}  \\
&=& \left( \frac{i}{m-1} \right)^{\sum_{j=0}^{j=k_m} 2 m_\alpha(n^{(j)}) - 1} 
\end{eqnarray*}
\noindent In conclusion, 
\begin{eqnarray*}
P[i = \max_l \{l : x_l \in tr(m,k_m) ] 
&=& \left( \frac{i}{m-1} \right)^{\sum_{j=0}^{j=k_m} 2 m_\alpha(n^{(j)}) - 1} \; 
\frac{(6 \alpha)}{(n-1)(n-2)} \;\; \sum_{j=0}^{j=k_m} \; (n^{(j)}-2) \\
&\leq & \frac{(6 \alpha)}{(n-1)(n-2)} \;\; \sum_{j=0}^{j=k_m} \; (n^{(j)}-2) \mbox{\hspace{1cm} as 
$i \leq (m-1)$} \\
&\rightarrow& 0 \mbox{\hspace{1cm} for large $n$ with respect to $k$ and $\alpha$} 
\end{eqnarray*}
\noindent Therefore, there is a negligible  probability for a
variable $x_i$ to be maximal in two or more trajectories $tr(m,k_m)$, as we can see
this event as the output of a binomial model with a very small
probability of success (``$x_i$ being maximal in some $tr(m,k_m)$''),
over $(n-i)(n-i-1)/2$ possible trajectories :
\begin{eqnarray*}
\mbox{Let } p = \max_{m,k_m} P[i = \max_l \{l : x_l \in tr(m,k_m)\} ]
\end{eqnarray*}
\begin{eqnarray*}
\mbox{Then, } P[\mbox{Two or more successes}]  \hspace{-45pt} && \\
&=& 1 - (P[\mbox{0 success}] + P[\mbox{1
  success}]) \\
&\leq& 1 - ([(1-p)^{\frac{(n-i)(n-i-1)}{2}}] + [\frac{(n-i)(n-i-1)}{2} p(1-p)^{\frac{(n-i)(n-i-1)}{2}-1}]) \\
&\rightarrow& 0 \mbox{\hspace{1cm} for large $n$,  as $p \rightarrow 0$
  for large $n$ with respect to $k$ and $\alpha$.} 
\end{eqnarray*}
The last thing to prove is that $k$ is not ${\cal O}(n)$ as $\alpha$
is a given constant.  Figure 3, which is computed with the theoretical
formula from (\ref{complexity_formula}), shows that the maximal value for $k$ is
negligible with respect to $n$ : $k \leq 30$ for $n=100.000$
and $\alpha = 4$, and $k \leq 55$ when $n=100.000$ and $\alpha=8$.
\end{proof}
\begin{figure}[htb]
\centering
\includegraphics[height=6cm]{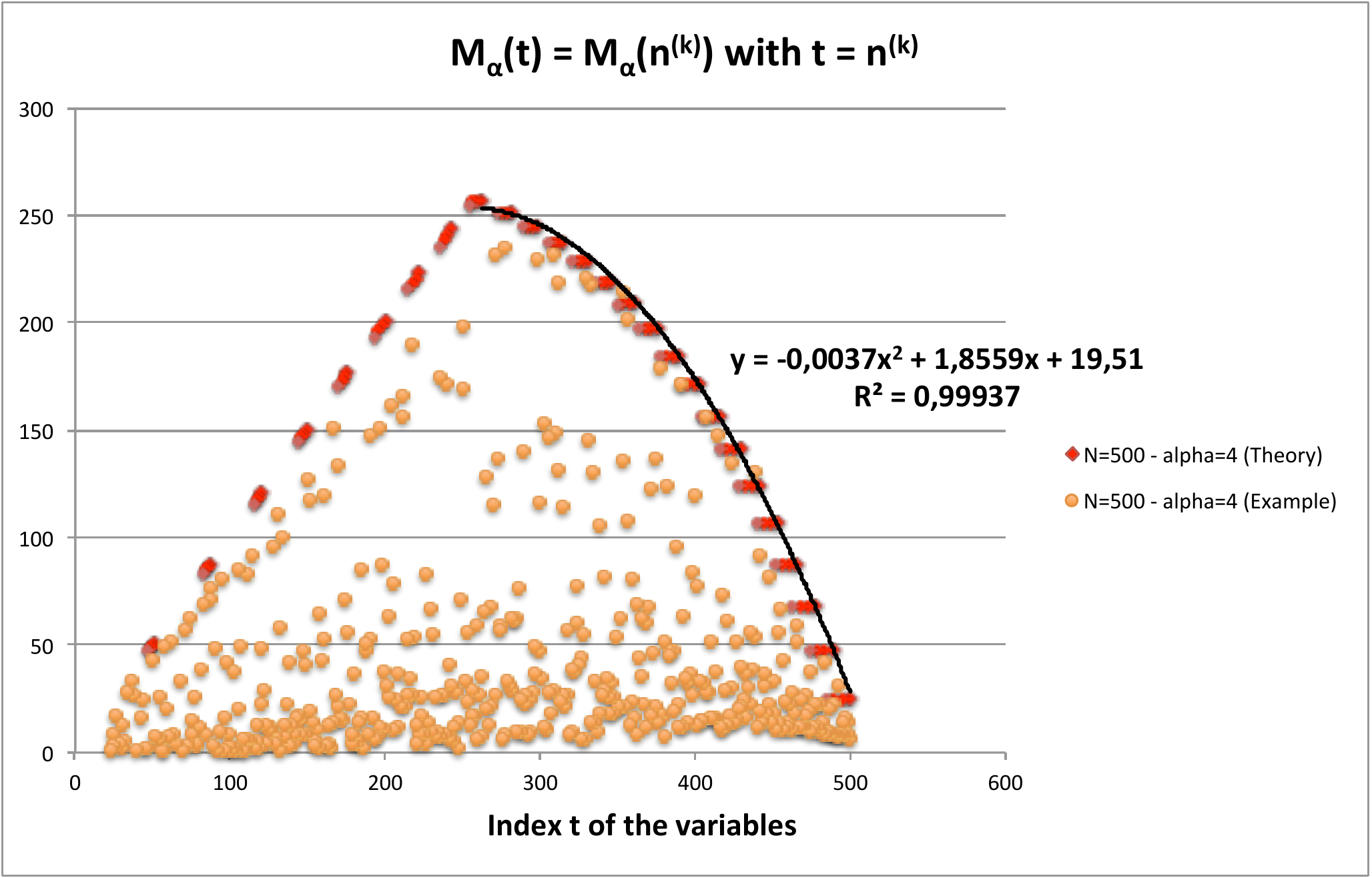}  \\
{\small Fig. 4 : Complexity wrt variable index $t$ : $M_\alpha(n^{(k)})=M_\alpha(t)$ for
  $n^{(0)} = n$. \\
The example comes from the Dimacs generator at https://toughsat.appspot.com/ }
\end{figure}  
\noindent {\bf Remark :} With some real generated 3-CNF-SAT problems, it is possible to observe
a ``cluster'' process, the size of one trajectory, i.e. the number of
$x_i$ involved in that trajectory, becoming more and
more important so that this trajectory attracts all the variables.
Then, $k$ is ${\cal}O(n)$, $P[x_i \in tr(n,k)] \rightarrow 1$ and the
complexity becomes exponential.  It is easy to solve these cases.  As
the variables are uniformly distributed in random 3-CNF-SAT problems,
each variable $x_i$ being
 repeated approximatively $3 \alpha$
times, it is possible to
permute joining variables $x_j$ (belonging to two or more
trajectories) with a smaller indexed variable,
such as $x_{j-1}$ (or $x_{j-2}$ if $x_{j-1}$ is already in a previous
trajectory, and so on).  The two
trajectories will then be dissociated. 
\noindent We propose the following {\it ``permutation'' algorithm} :
{\small \textsf{ 
\begin{itemize}
\item[$\diamond$] First, apply the sorting algorithm over the 3-CNF-SAT
  problem;
\item[$\diamond$] Sort each clause $[\neg]x_r \vee [\neg]x_s \vee
  [\neg] x_t$ so that $r \ge s \ge t$;
\item[$\diamond$] Beginning with the last clause (with $[\neg]x_n$),
  mark $x_j$ where $j=\max \{i : x_i \in
  V(x_n) \}$ as already belonging in a trajectory and initialize
  $W(x_n):=V(x_n)$ and $W(x_j):=V(x_n)$ where
  $W(x_j)\equiv \cup_i V(x_i)$ for $i$ such that $x_i \in
  tr(\cdots,x_j)$;
\item[$\diamond$] Loop over $k:=1$ to $k:=n-3$ with clauses having
  $[\neg]x_{n-k} $ as the highest indexed variable;  if $x_{n-k}$ is
  already marked as belonging to a trajectory, do
  $W(x_{n-k}):=W(x_{n-k}) \cup V(x_{n-k})$ otherwise
  initialize $W(x_{n-k}):=V(x_{n-k})$;
\item[] Consider $x_j$ where $j=\max \{i : x_i \in
  W(x_{n-k}) \}$; do while ($x_j$ is already marked as belonging in a
  trajectory and $j \ge 3 \alpha$)  relabel $x_j \leftrightarrow
  x_{j-1}$ and $j:=j-1$ 
\item[] [we do not consider $j < 3 \alpha$ as merging
  of trajectories for small indexes is not a problem because $M_\alpha(j)=j$];
\item[] Initialize $W(x_j):=W(x_{n-k})$. 
\end{itemize}
}  }
\noindent Figure 4 shows the result for a
Dimacs generated 3-CNF-SAT problem with 500 variables and $\alpha =
4$.  We apply the sorting and the permuting algorithms on the
generated file to eliminate joining trajectories. \\[12pt]
\noindent If we have proved in this section that the complexity is bounded with
respect to $n$, we still have to show that complexity is not
increasing with respect to $\alpha$, which is not the case for $M_\alpha(t)$. \\[12pt]
\section{Complexity analysis with respect to $\alpha$}
\noindent It is easy to see that the complexity is
an increasing function of $\alpha$ for exact uniformly distributed $\alpha$-random 3-CNF-SAT
problems, at least for small $\alpha$, as smaller $\alpha$-random 3-CNF-SAT problems can be viewed
as subsets of larger $\alpha$-random problems.  \\[12pt]
\noindent But there should be somewhere {\bf \em a threshold for $\alpha$} as large
$\alpha$-random problems are easy to solve because unsatisfiability
is often a consequence of a subset of the problem.  Empirical
results from the literature suggest that this threshold for $\alpha$
is $\approx 4.258$.  See \cite{Crawford199631}. \\[12pt]
\noindent The analysis of complexity with respect to $\alpha$ will be
done through ${\cal S}_{\varphi}$, the set of all satisfying solutions for the
3-CNF-SAT problem $\varphi$.  See definition
(\ref{no2}). \\
\begin{theorem} {\bf For exact uniformly distributed $\alpha$-random
  3-CNF-SAT problems $\varphi=\{\psi_j\}_{1 \leq j \leq m}$ with
  $n$ variables and $m$ clauses, we get for large $n$ and $m$ the following
  expected number of solutions  :}
\begin{eqnarray}
E[ \# {\cal S}_{\varphi} ] = E[ \#\{(x_1, \cdots,x_n) \in \{0,1\}^n|\varphi(x_1,
\cdots,x_n) = 1\} ] =  7 \; (\frac{7}{4})^{n-3} \; (\frac{7}{8})^{(m-n+2)}
\end{eqnarray}
\end{theorem}
\begin{proof}
\mbox{}\\
\noindent $\bullet$ Let us {\it re-order} the $m$ clauses $\psi_j$
 in such a way that each clause has only {\it one} new additional
variables with respect to the set of variables appearing in the
previous clauses. \\
\noindent $\bullet$ Let ${\cal V}_k=\{x_i \; | \; \exists j \;,\; 1 \leq j \leq k
\;:\; x_i \mbox{ appears in } \psi_j \}$.  The {\it re-ordering} of the clauses yields to embedded
  subsets ${\cal V}_1 \subseteq {\cal V}_2 \cdots \subseteq {\cal V}_{n-2}
  = \{x_1, \cdots, x_n\}$ with 
$\# {\cal V}_1 =  3, \cdots, \#{\cal V}_k = k+2, \cdots,$ \linebreak $ \#{\cal V}_{n-2} =
n$ and $\# {\cal V}_{k'} =  n \;\; \forall \; k' \ge n-2.$\\
The cases where all clause $\psi_{k+1}$ introduces {\it two or three} new
additional variables to ${\cal V}_{k}$ are to be neglected, as this means that the
3-CNF-SAT problem can be split into two sub-problems with one
or zero common variable, which reduces drastically the complexity of the
problem.  \\
\noindent $\bullet$ Let us look at the expected effect of a clause $\psi_j \; (1
\leq j \leq m)$ over
the number of solutions : 
\begin{enumerate}
\item \noindent Let us consider $\psi_1$.  \\
The first clause yields to
  $7 \cdot 2^{n-3}$ solutions.  The {\it matrix representation} of
  $\psi_1$ will be a $7 \times 3$ matrix.
\item \noindent Let us consider $\psi_2$. \\
Let $\psi_2$ introduces only one new additional
  variable $x_t$, and let $x_r$ and $x_s$ be the two common variables for
  $\psi_1$ and $\psi_2$ : 
\begin{eqnarray*}
[\psi_1]=  \left( \begin{array}{ccc}
x_q &  x_r &  x_s \\
\hline
\multicolumn{3}{c}{\mbox{7 lines}}
\end{array} \right) \mbox{  and  } 
[\psi_2]=  \left( \begin{array}{ccc}
 x_r &  x_s & x_t \\
\hline
\multicolumn{3}{c}{\mbox{7 lines}}
\end{array} \right) 
\end{eqnarray*}
Depending on the sign of the literal $x_t$ in $\psi_2$, the result
matrix for $[\psi_1 \wedge \psi_2]$ will get the $7$ lines of
$[\psi_1]$ with a {\it zero} in the column for $x_t$ if $\psi_2 = [\neg] x_r \vee
[\neg] x_s \vee \neg x_t$ or with a {\it one} if $\psi_2 = [\neg] x_r \vee
[\neg] x_s \vee  x_t$.  This corresponds to solutions where the
literal $[\neg] x_t$ is satisfied.  \\
On the contrary, when the value
in the column for $x_t$ is opposite to the sign of $[\neg] x_t$, the
satisfiability of $\psi_2$ should pass through the literals $x_r$
and $x_s$.  Among the four possible values for $(x_r,
  x_s)$, only three will be accepted.  One couple for $(x_r, x_s)$ will be
  ruled out, as well in matrix $[\psi_1]$ as in $[\psi_2]$.  This corresponds to
  one or two lines deleted in $[\psi_1]$, depending on the sign for $x_r$ and
  $x_s$ in $\psi_1$.  The expected number of lines deleted in
  $[\psi_1]$ will be : $ 1 \cdot  P$[one deletion$] + 2 \cdot P[$two
  deletions$] = 1 \cdot \frac{1}{4} + 2 \cdot \frac{3}{4} =
  \frac{7}{4}$.  Therefore, the expected number of lines in
  $[\psi_1 \wedge \psi_2]$ will be equal to $7 \; + \; (7 - \frac{7}{4}) = 7
  (1 + \frac{3}{4}) = 7 (\frac{7}{4}) = 12,25.$  And the boundaries for $\# [\psi_1 \wedge
  \psi_2]$ are $[\min_2,\max_2] = [12,13]$.  
\item \noindent Let us consider  $\psi_3$. \\
Let $\psi_3$ introduce a new additional
  variable.  Using the same type of arguments as for $\psi_2$,
  the expected number of deleted lines in $[\psi_1 \wedge \psi_2]$ will
  be equal to : 
\begin{eqnarray*}
E(\# \; \mbox{deletions in } \; \psi_1 \wedge  \psi_2) &=& \\
& & \hspace{-2cm} \sum_{k=12}^{13} \; \left( \sum_{d=1}^{2^2} \; d \cdot P[ d 
 \mbox{ deletions } | \# [\psi_1 \wedge \psi_2] =  k ] \right) \cdot P[\# [\psi_1 \wedge \psi_2] = k ] \\
\end{eqnarray*}
Let us consider here an example where $ \# [\psi_1 \wedge \psi_2] = 13$.
\begin{eqnarray*}
\mbox{For instance, } [\varphi] = [(x_1 \vee x_2 \vee \neg x_3) \wedge (x_2 \vee \neg x_3
\vee \neg x_4)] =
 \left( \begin{array}{cccc}
x_1 & x_2 & x_3 & x_4\\
\hline
0 & 0 & 0 & 0 \\
0 & 1 & 0 & 0 \\
0 & 1 & 1 & 0 \\
1 & 0 & 0 & 0 \\
1 & 0 & 1 & 0 \\
1 & 1 & 0 & 0 \\
1 & 1 & 1 & 0 \\
0 & 0 & 0 & 1 \\
0 & 1 & 0 & 1 \\
0 & 1 & 1 & 1\\
1 & 0 & 0 & 1 \\
1 & 1 & 0 & 1 \\
1 & 1 & 1 & 1 \\
\end{array} \right)  
\end{eqnarray*}
Let us consider the couples $(x_i,x_j)$ and the number of
deleted lines for each case :
\begin{eqnarray*}
\begin{array}{ccccccc}
x_i&x_j&\#\mbox{ del.}&&x_1&x_2&\#\mbox{ del.}\\
0&0&d_1&&0&0&2 \\ 
0&1&d_2&\mbox{\it For the above example : }&0&1&4 \\ 
1&0&d_3&&1&0&3 \\ 
1&1&d_4&&1&1&4 \\
\end{array}
\end{eqnarray*} 
We see that, whatever the value of $ \# [\psi_1 \wedge \psi_2]$,
$\sum_{i=1}^4 d_i = \# [\psi_1 \wedge
\psi_2].$ 
So, the expected number of deleted clauses, independently from the
case $(x_i,x_j)$, will be :
\[ d_1 \cdot \frac{1}{4} + d_2 \cdot \frac{1}{4} + d_3 \cdot \frac{1}{4}
+ d_4 \cdot \frac{1}{4} = \frac{\sum_i d_i}{4} = \frac{\# [\psi_1 \wedge
  \psi_2]}{4} \]
Therefore, the expected number of lines in  $[\bigwedge_{i=1}^3
\psi_i]$ will be :
\begin{eqnarray*}
E[ \# [\bigwedge_{i=1}^3 \psi_i] ] 
&=& \sum_{k} \{ \# [\psi_1 \wedge \psi_2] + (\# [\psi_1 \wedge
\psi_2] - \# \mbox{ deletions } )\} \cdot
P[\# [\psi_1 \wedge \psi_2] = k ] 
\\
&=& \sum_{k=12}^{13} \{ \# [\psi_1 \wedge \psi_2] + (\# [\psi_1 \wedge
\psi_2] - \frac{ \# [\psi_1 \wedge
\psi_2]}{4} )\} \cdot
P[\# [\psi_1 \wedge \psi_2] = k ] 
\\
&=& \left( 12 \{1 + (1 -
    \frac{1}{4})\} \cdot \frac{3}{4} \right)
+ \left( 13 \{1 + (1 -
    \frac{1}{4})\} \cdot \frac{1}{4} \right) \\
 &=& E[ \# [\bigwedge_{i=1}^2 \psi_i] ] \cdot \frac{7}{4} \\
&=& 7 \cdot (\frac{7}{4})^{2}\\
&=& 21,4375
\end{eqnarray*}
And $\# [\bigwedge_{i=1}^3 \psi_i] \in [\min_3,\max_3] = [12+12-4,13+13-1] = [20,25]$
\item \noindent Let us now consider a given clause $\psi_j \;\; (j \leq n-2)$. \\
We know that $\psi_j$ introduces a new additional
  variable.  Then, using the same type of arguments as for $\psi_3$,
  the expected number of lines will be :
 \begin{eqnarray*}
 E[ \# [\bigwedge_{i=1}^j \psi_i] ] &=& E[ \# [\bigwedge_{i=1}^{j-1}
 \psi_i] ] \cdot  (\frac{7}{4}) = 7 \cdot (\frac{7}{4})^{j-1}. \\
\mbox{and } \;\;\; \# [\bigwedge_{i=1}^j \psi_i] &\in& [2 \; \mbox{$\min$}_{j-1}-2^{j-1}, 2
\; \mbox{$\max$}_{j-1}-1] \\
 &\in& [\max \{0,2^{j-1} (8-j)\} \;,\; 6 \cdot 2^{j-1} + 1] 
\end{eqnarray*}
\item So, for $\psi_{n-2}$, we have ($n \geq 10$) :
\begin{eqnarray}
E[ \# [\bigwedge_{i=1}^{n-2} \psi_i] ] &=& 7 \cdot (\frac{7}{4})^{n-3} \\
\mbox{and } \;\;\; \# [\bigwedge_{i=1}^{n-2} \psi_i] &\in& [0 \;,\; 6
\cdot 2^{n-3} + 1] 
\end{eqnarray} 
\item For $\psi_j$ where $j > n-2$, no new variable will be added, and
  the number of solution will only decrease. \\
Using the same previous argument, we can consider the six possible cases  $(x_i,x_j,x_k)$
and the corresponding $d_i$ with $1 \leq i \leq 8$.  Here again, we
get that : 
\[\sum_{i=1}^8 d_i = \#  [\bigwedge_{i=1}^{j-1} \psi_i]\]
Thus, the expected  number of deleted lines in $[\bigwedge_{i=1}^{j-1}
\psi_i]$ will be $\frac{\#  [\bigwedge_{i=1}^{j-1} \psi_i]}{8}.$\\
There is no other operation to do for $\psi_j$.  We only have to delete
some lines in $[\bigwedge_{i=1}^{j-1} \psi_i]$.  So, computing the
remaining lines, we get :
\begin{eqnarray*}
E[ \# [\bigwedge_{i=1}^{j} \psi_i] ] = E[ \# [\bigwedge_{i=1}^{j-1}
\psi_i] ] \cdot \frac{7}{8} 
\end{eqnarray*}  
And the boundaries will be : 
\[ \# [\bigwedge_{i=1}^{j} \psi_i] \in [0, \mbox{$\max$}_{j-1} - 1] \]
\item Finally, for the last clause $\psi_m$, we get :
\begin{eqnarray*}
E[ \# {\cal S}_\varphi ] = E[ \# [\bigwedge_{i=1}^{m} \psi_i] ] &=& 7 \;
(\frac{7}{4})^{n-3} \; (\frac{7}{8})^{(m-n+2)} \\[12pt]
\mbox{and } \;\;\; \# [\bigwedge_{i=1}^{m} \psi_i] &\in& [0 \;,\; 6
\cdot 2^{n-3} - m + n - 1]
\end{eqnarray*}  
\end{enumerate}
\end{proof}
\begin{theorem} {\bf The most difficult {\em exact uniformly distributed $\alpha$-random
  3-CNF-SAT 
problems} are the ones with a
 ratio $\alpha = \frac{m}{n}$ approximately equal to $ 5,19 $.}
\end{theorem}
\begin{proof}
Let us consider {\em exact uniformly distributed $\alpha$-random
  3-CNF-SAT problems.}  The most difficult problems are the ones where the decision between
satisfiability and unsatisfiability arises only when
considering the last clause $\psi_m$. 
This is equivalent to have $E[\# {\cal S}_{\varphi}] \approx 1$.  We
get : 
\begin{eqnarray*} 
E[\# {\cal S}_{\varphi}] \approx 1 &\Leftrightarrow& 7 (\frac{7}{4})^{n-3} (\frac{7}{8})^{m-n+2}
\approx 1 \\
&\Leftrightarrow&  (\frac{7}{8})^{m} \; 2^{n} \approx 1 \\
&\Leftrightarrow&  m \; \log(\frac{7}{8}) + n \; log(2) \approx 0 \\
&\Leftrightarrow&  (\alpha \; n) \; \log(\frac{7}{8}) + n \; log(2) \approx 0 \\
&\Leftrightarrow&  \alpha \; \log(\frac{7}{8}) \approx - log(2) \\
&\Leftrightarrow& \alpha \approx  \frac{- \log(2)}{\log(\frac{7}{8})} \\ 
&\Leftrightarrow& \alpha \approx  5,19089307  
\end{eqnarray*}
\end{proof}
\begin{theorem}
{\bf The relation between {\em exact uniformly distributed  $\alpha$-random
    3-CNF-SAT problems} and {\em usual $\alpha$-random
  3-CNF-SAT problems} can be seen as a reduction of the ratio $\alpha$
through the function : $3 \alpha - \sqrt{1,9098 \;
  \alpha}$}.
\end{theorem}
\begin{proof}
\mbox{}\\
\noindent When considering {\em exact uniformly distributed $\alpha$-random
3-CNF-SAT problems}, each
literal occurs with exactly the same frequency in the $m$ clauses, only
the combination of the literals in each clause being random. We have : $\# x_i
= 3 \; \alpha$. \\
\noindent But the {\em usual uniform $\alpha$-random 3-CNF-SAT
  problems} are such that : $E[ \# x_i ] = 3 \; \alpha$,
where the variables are drawn randomly from a {\it multinomial}
population with $P[x_i$ appears in a clause$] = p_i = \frac{3 \;
  \alpha}{m} = \frac{3}{n}$.
 For large $n$, the number of occurrence for each variable will asymptotically follow a
{\it Normal distribution} $N(\mu,\sigma^2)$ with $\mu = m \cdot p_i =
3 \; \alpha$ and $\sigma^2 = m \cdot p_i (1-p_i) \approx 3 \; \alpha$. \\
\noindent If we consider, {\it after
sorting the clauses as explained in our descriptor approach}, the second
half of the clauses (where $M_\alpha (t) \geq t$), we will get a {\it
  folded normal distribution} for \linebreak $D_{i \;\; \mbox{\it [usual
    $\alpha$-random 3-CNF]}} = |\#\{x_i\} - E[\#\{x_i\}]| =
|\#\{x_i\} - 3 \; \alpha|$.  We have :
\begin{eqnarray*}
D_i &\sim& |N(0, 3 \; \alpha)| \\
E[D_i] &\approx& \sigma \; \sqrt{\frac{2}{\pi}} \\
&\approx& \sqrt{\frac{6 \; \alpha}{\pi}} = \sqrt{1,9098 \; \alpha} \\
\Rightarrow \; E[\# \{x_i\}] &\approx& 3 \alpha - \sqrt{1,9098 \;
  \alpha} \; \mbox{ for the clauses where $M_\alpha (t) \geq t$ }
\end{eqnarray*}
\noindent So, if we have $\# \{x_i\} = \alpha$ in the {\em exact
  uniformly distributed $\alpha$-random 3-CNF-SAT problems}, this
corresponds to an {\em ``folded''} expected occurency $E[\# \{x_i\}] \approx 3 \alpha - \sqrt{1,9098 \;
  \alpha} $ for {\em usual uniform $\alpha$-random 3-CNF-SAT
  problems}.  
\end{proof}
\mbox{}\\
\begin{corollary} The threshold $\alpha = 5,19$ found for exact
    uniformly distributed $\alpha$-random 3-CNF-SAT problems
  corresponds approximatively to a reduced threshod 
\[ \alpha =
  \frac{E[\#\{x_i\}]}{3} = 
  5,19 - \frac{\sqrt{1,9 \times 5,19}}{3} = 4,14135 \]
for usual uniform $\alpha$-random 3-CNF-SAT problems.
\end{corollary}%
\begin{figure}[htb]
\centering
\includegraphics[height=6cm]{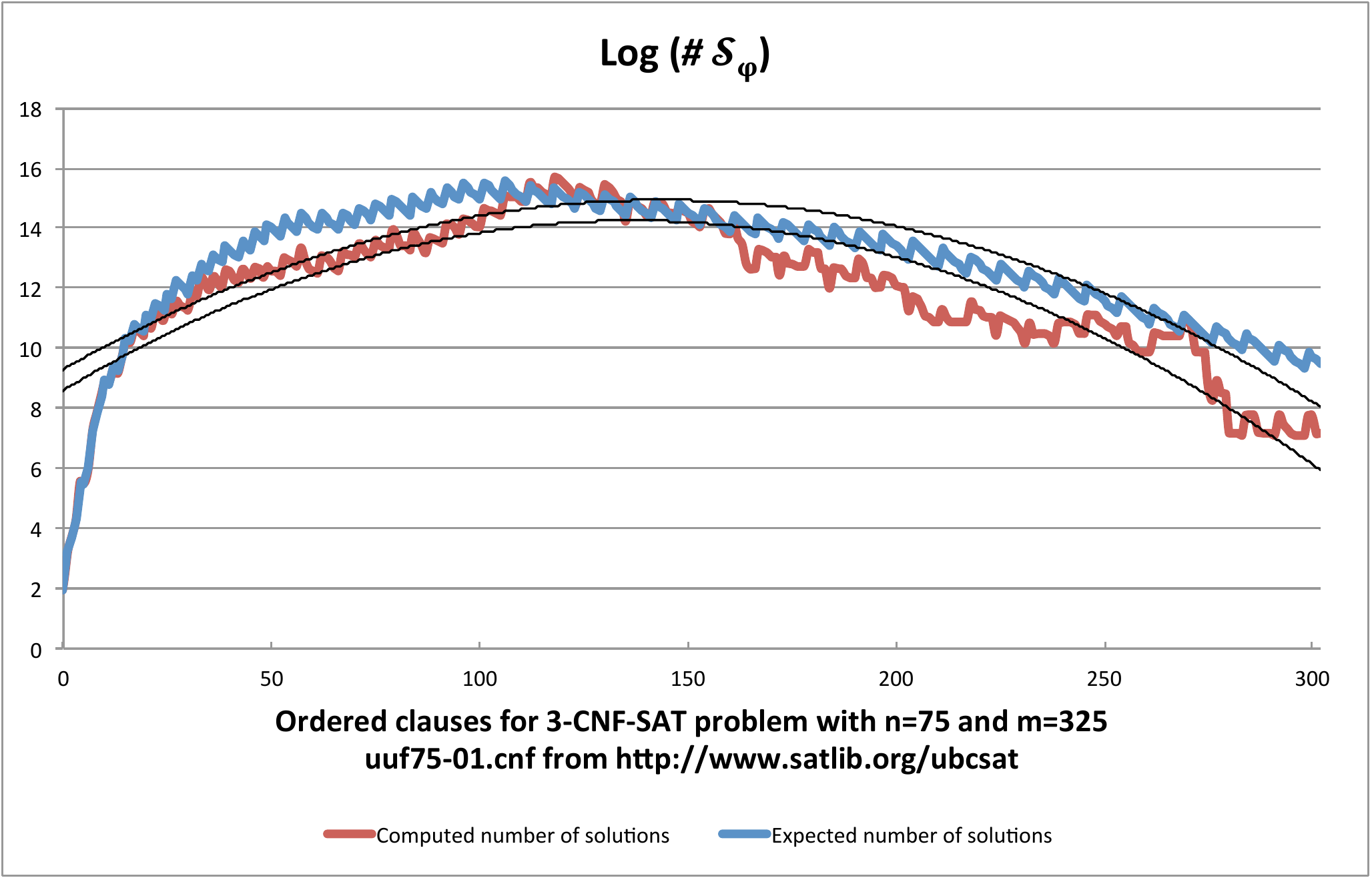}  \\
{\small Fig. 5 : Number of solutions with respect to the analyzed
  clauses\\ for a 3-CNF-SAT problem with $n=75$ and $m=325$}
\end{figure}  
\mbox{}\\
\noindent  {\it Note : } This is still a theoretical value for the threshold.
Indeed, for usual uniform $\alpha$-random generated 3-CNF-SAT problem, we detect a
small difference between the observed and the theoretical expected number of solutions with respect of
the first $j$ analyzed clauses $\bigwedge_{i=1}^{j} \psi_i$.  The theoretical expected number of
  solutions $E[ \# [\bigwedge_{i=1}^j \psi_i] ]$ is defined as $7 \cdot
  (\frac{7}{4})^{s} \cdot (\frac{7}{8})^{t}$, where $s$ is the number
  of clauses in $\{\psi_2,\cdots,\psi_j\}$ introducing new additional variable and $t$ the number
  of remaining clauses. Figure 5
shows the situation for a 3-CNF-SAT
problem with $n=75$ and $m=325$, taken from {\it
  http://www.satlib.org/ubcsat}. \\
This difference shows that theoretical expected
values are over-estimating the observed values. Let us note that 
we {\it re-ordered} the $m$ clauses $\psi_j$ in such a way that 
new additional variables are appearing as lately as possible in the 3-CNF-SAT
problem (in order not to reach too large numbers for $\# {\cal S}_\varphi$). \\
\section{Conclusions and future researches}
\noindent Our researches were built on {\it exact uniformly
  distributed $\alpha$-random 3-CNF-SAT problems}.  The complexity
analysis was mainly done in terms of {\it expected value} for some
characteristics.  We see that these expected values are
over-estimating the real values.  This means that our conclusions
about the most difficult value for $\alpha \; [= 5,19]$, and therefore about the
maximum theoretical value for the complexity $2^{M_\alpha(t)} \; [= 2^{490}]$ are
over fitted.  Future researches will try to
suppress this bias to be more accurate in our estimation of the
complexity for the \classNP problems. \\[12pt]
We have seen that for $\alpha \approx 5,19$, the maximum complexity for a
  $\alpha$-random 3-CNF-SAT problem will be around $2^{490}$ whatever
  the number of variables.  {\bf The \classNP problems are then not
  exponential but {\em bounded exponential} problems.  This makes them
  belonging to
  \classP.} 
 But even with ``yottaflops'' computers ($10^{24}$
instructions by second), this can take about 
``$10136575708788609985206606935908809922268405942697014391424964252461$
$88692463039064879247034987638184445605903560477$''
centuries to solve such problems. $\ddot \smile$   This is not
exponential, only a huge constant upper bound. \\[12pt]
\noindent Even if this paper is mostly theoretical, each theorem was
validated by extensive numerical tests.  Future researches will be to
improve our different algorithms implementing the descriptor approach
for 3-CNF-SAT problems\footnote{I would like to thank Dr. Johan
  Barth\'elemy
for his help in terms of writing and testing these algorithms, as
well as the SMART department of the University of Wollongong for their
welcome.  Codes will be available on www.github.com}. 
%
%
%
%
\nocite{Sipser92}

\addcontentsline{toc}{part}{\mbox{Bibliography}}  
\bibliographystyle{plain}   
\bibliography{mybib}   

\newcommand{\noopsort}[1]{} \newcommand{\printfirst}[2]{#1}
  \newcommand{\singleletter}[1]{#1} \newcommand{\switchargs}[2]{#2#1}
\begin{thebibliography}{1}

\bibitem{burris2012a}
Stanley Burris.
\newblock {\em A Course in Universal Algebra}.
\newblock Dover Pubns, City, 2012.

\bibitem{cormen2001}
Th. Cormen, Ch. Leiserson, R.~Rivest, and Cl. Stein.
\newblock {\em Introduction to Algoritmics}.
\newblock MIT Press, Cambridge, 2nd edition, 2001.

\bibitem{Crawford199631}
James~M. Crawford and Larry~D. Auton.
\newblock Experimental results on the crossover point in random 3-sat.
\newblock {\em Artificial Intelligence}, 81(1–2):31 -- 57, 1996.
\newblock Frontiers in Problem Solving: Phase Transitions and Complexity.

\bibitem{arxiv2}
R{\'{e}}mon Marcel.
\newblock About the impossibility to prove p != np or p = np and the
  pseudo-randomness in np.
\newblock Published in Arxiv : http://arxiv.org/abs/0904.0698v2, January 2010.

\bibitem{Sipser92}
M.~Sipser.
\newblock \mbox{The History and Status of the P versus NP Question}.
\newblock {\em Proceedings of the 24th Annual Meeting ACM}, pages 603--618,
  1992.

\end{thebibliography}
\end{document}